\documentclass[letterpaper,11pt]{article}

\usepackage{amsmath}
\usepackage{amssymb}
\usepackage{indentfirst}
\usepackage{hyperref}
\usepackage{amscd}
\usepackage{pifont}
\usepackage{amsthm}
\usepackage{charter}

\numberwithin{equation}{section}
\allowdisplaybreaks

\newtheorem{theorem}{Theorem}[section]
\newtheorem{lemma}[theorem]{Lemma}
\newtheorem{proposition}[theorem]{Proposition}
\newtheorem{corollary}[theorem]{Corollary}
\newtheorem{remark}[theorem]{Remark}
\newtheorem{definition}[theorem]{Definition}
\newtheorem{example}[theorem]{Example}

\title{Deformations of Courant Algebroids and Dirac Structures via Blended Structures}
\author{Xiang Ji\\The Pennsylvania State University, New Kensington}
\date{}

\begin{document}
\maketitle


\begin{abstract}
Deformations of a Courant Algebroid $(E,\langle\cdot,\cdot\rangle,\circ,\rho)$ and its Dirac subbundle $A$ have been widely considered under the assumption that the pseudo-Euclidean metric $\langle\cdot,\cdot\rangle$ is fixed. In this paper, we attack the same problem in a setting that allows $\langle\cdot,\cdot\rangle$ to  deform. Thanks to Roytenberg, a Courant  Algebroid is equivalent to a symplectic graded $Q$-manifold of degree $2$. From this viewpoint, we extend the notions of graded $Q$-manifold, DGLA and $L_\infty$-algebra all to ``blended'' version so that Poisson manifold, Lie algebroid, Courant algebroid are unified as blended $Q$-manifolds; and define a submaniold $\mathcal{A}$ of ``coisotropic type" which naturally generalizes the concepts of coisotropic submanifolds, Lie subalgebroids and Dirac subbundles.  It turns out the deformations a blended homological vector field $Q$ is controlled by a blended DGLA, and the deformations of $\mathcal{A}$ is controlled by a blended $L_\infty$-algebra. The results apply to the deformations of a Courant algebroid and
its Dirac structures, the deformations of a Poisson manifold and its coisotropic
submanifold, the deformations of a Lie algebroid and its Lie subalgebroid.
\end{abstract}

\section{Introduction}
\par Deformation theories are traditionally described using differential graded Lie algebras (DGLA, for short), with the goal of establishing the 1-1 correspondence between the deformations of a given object and the Maurer-Cartan elements of an associated DGLA. However, some deformation problems require an extension of the concept so that the failure in describing the deformations by Maurer-Cartan elements of a DGLA is measured by an infinite sequence of brackets that are intrinsically compatible. Such a notion is called an $L_\infty$-algebra, and nowadays it is widely adopted for deformations. An $L_\infty$-algebra $L$ is said to control the deformations of an object if the deformations are exactly those solutions of the Maurer-Cartan equation of $L$. Therefore, given a deformation problem in consideration, a routine task is to construct an $L_\infty$-algebra and show it controls the deformations.

\par The motivating problem of this paper is the deformation problems of a Courant algebroid and its Dirac structures.
Courant bracket was first studied by T. Courant (\cite{c}), and was later formalized by Z. Liu, A. Weinstein and P. Xu to study the structure of Lie bialgebroids. Roughly speaking, a Courant algebroid is a vector bundle $E$ equipped with a pseudo-inner product $\langle\cdot,\cdot\rangle$, an anchor $\rho$ from $E$ to the tangent bundle $TM$ of the base manifold $M$, and a bracket $\circ$ on the space $\Gamma(E)$ that fails to be a Lie bracket. The failure is leveraged between the loss of antisymmetry and the loss of Jacobi identity. A Dirac structure is an involutive Lagrangian subbunle $A$ of $E$. It lends great power in unifying geometric structures such as symplectic, Poisson, complex and generalized complex structures. 

\par Due to the wide applications of Courant algebroid and Dirac structure, their deformation problems have been broadly considered and kept a topic attracting interests. In \cite{lwx}, Z. Liu, A. Weinstein and P. Xu first construct a DGLA $(\Omega_A,d_A,[\cdot,\cdot]_{L})$ associated to the Dirac structure $A$ assuming the existence of a transversal Dirac structure $L$, and use it to control the small deformations of $A$. Later in \cite{kw1}, formal deformations of a Dirac structure are considered. Concerning Courant algebroid, from pure algebraic viewpoint, the deformations are described via Poisson algebra in \cite{kw2} under the assumption that the pseudo-inner product is fixed. With the same assumption,  in \cite{fz}, via D. Roytenburg's derived brackets viewpoint, a DGLA is constructed to control the deformations of the bracket $\circ$ and the anchor $\rho$; upon the selection of a complement $B$, an $L_\infty$-algebra $V_{A,B}$ is constructed to control the deformations of $A$; the two are combined to control their simultaneous deformations. Very recently, the uniqueness of the DGLA $(\Omega_A,d_A,[\cdot,\cdot]_{L})$ in \cite{lwx} and the $L_\infty$-algebra $V_{A,B}$ in \cite{fz} is proved in \cite{gms}. To be specific, a different selection of the complement results in an $L_\infty$-isomorphic DGLA or $L_\infty$-algebra.

\par In this paper, we attempt to attack the deformation problem of Courant algebroid and Dirac structure in the most general setup: 
\par 1) the pseudo-inner product $\langle\cdot,\cdot\rangle$ is deformed with the bracket $\circ$ and the anchor $\rho$;
\par 2) a deformed subbundle of $A$ is not assumed to be Lagrangian in advance.\\
Thanks to D. Roytenberg (\cite{r2}), Courant algebroids are in 1-1 correspondence with symplectic graded $Q$-manifolds of degree $2$, i.e. a graded vector bundle $\mathcal{E}$ equipped with a symplectic structure $\Omega$ of degree $2$ and an involutive vector field $X_Q$ of degree $1$ that preserves $\Omega$. Considering the Poisson structure $\pi$ inverse to $\Omega$, by carefully selecting the degree on multi-vector fields, $\pi$ is of degree $-1$. The pair $(\pi,X_Q)$ is equivalent to the Courant algebroid structure $(\langle\cdot,\cdot\rangle,\circ,\rho)$ on $E$ and satisfies
\[
\begin{cases}
	[\pi,\pi]_\text{SN} = 0,\\
    [\pi,X_Q]_\text{SN} = 0,\\
    [X_Q,X_Q]_\text{SN} = 0.
\end{cases}
\]
Here $[\cdot,\cdot]_\text{SN}$ is the Schouten-Nijenhuis bracket of multi-vector fields on graded manifolds. Let $\mathfrak{X}(\mathcal{E})$ be the collection of multi-vector fields on $\mathcal{E}$. It is clear that the operator $[X_Q,\cdot]$ is a differential of the DGLA $(\mathfrak{X}(\mathcal{E}),[\cdot,\cdot]_\text{SN})$. However, although the composition of the operator $[\pi,\cdot]$ is zero, it is not a differential since its degree is $-1$. In order to combine them together, we attempt to add $\pi$ and $X_Q$ into a blended structure $\pi+X_Q$. Following this idea, the notion of DGLA is extended to a 'blended' version, which is then used to control the deformations of Courant algebroid.

\par On the other hand, the vector field $\pi+X_Q$ satisfying $[\pi+X_Q,\pi+X_Q]_\text{SN}=0$, which is analogous to a homological vector field. With this idea in mind, we immediately notice a list of geometric structures that fall into the same scheme.
\begin{center}
\begin{tabular}{c|c|c}
Geometric structure & Graded structure & Characteristic submanifold\\
\hline\hline
Poisson manifold & ? & coisotropic submanifold\\
\hline
Lie algebroid & $Q$-manifold of degree $1$ & Lie subalgebroid\\
\hline
Courant algebroid & symplectic $Q$-manifold of degree $2$ & Dirac structure\\
\hline
VB-Courant algebroid(\cite{l}) & $Q$-manifold of degree $2$ & ?\\
\hline
\end{tabular}
\end{center}
It seems one can unify all these structures by properly generalizing the $Q$-structure. Considering that Poisson structure is a bi-vector field, we first allow homological vector field to be a multi-vector field. Secondly, in Courant algebroid the Poisson structure is of degree $-1$, we continue to extend the notion so that a homological vector field contains components of both degree $1$ and $-1$. Such a homological vector field is called blended, and consequently, the geometric spaces in the above table are unified to blended $Q$-manifolds.

\par It follows that the characteristic submanifolds such as coisotropic submanifold, Lie subalgebroid and Dirac structures may have a uniform definition. Motived by the definition of a coisotropic submanifold, both Lie subalgebroid and Dirac structure satisfy that the blended homological vector field is annihilated by the conormal bundle of the submanifold. Using this condition, we define submanifolds ``of coisotropic type" in a blended $Q$-manifold, and unify all the characteristic submanifolds in the table. This also enlightens us that the deformation problems of those submanifolds could be solved in the general setup once for all.

\par Following the construction of an $L_\infty$-algebra associated with a Lie-subalgebroid, a `blended'-$L_\infty$-algebra $L_\mathcal{A}$ associated with a submanifold $\mathcal{A}$ of ``coisotropic type'' is constructed similarly. The structure maps of $L_\mathcal{A}$ are given by the higher derived brackets of the homological vector field $Q$, and under necessary regularity conditions, if a deformation of $\mathcal{A}$ is of coisotropic type then it is a Mauer-Cartan element of $L$. The deformations of coisotropic submanifolds, Lie subalgebroids and Dirac structures can be solved as applications.

\par The structure of the paper is as follows. We extend the notion of DGLA and $L_\infty$-algebra to blended version in Section 1, and extend the notion of graded $Q$-manifold to blended $Q$-manifold in Section 2 with submanifolds of ``coisotropic type'' defined. In Section 3, the deformations of a general blended $Q$-structure and the deformations of a submanifold of coisotropic type are attacked. It turns out that the former deformations are controlled by a blended DGLA, while the latter ones are controlled by a blended $L_\infty$-algebra. When the two structures are not ``blended'', they can be combined to control simultaneous deformations. Finally, in Section 4, the deformation problems of a Courant algebroid and its Dirac structure are solved as applications of the results in Section 3.

\section{Blended $L_\infty$-algebras}
\par For the purpose of describing our results on deformations in latter sections, we extend the notions of differential graded Lie algebra (DGLA) and $L_\infty$-algebra both to some blended version. We also show that Voronov's higher derived brackets construction of $L_\infty$-algebra applies directly to this blended version.

\par A $\mathbb{Z}$-graded vector space $V$ over a field $\mathsf{k}$ is a direct sum $\oplus_{k\in\mathbb{Z}}V_k$ of $\mathsf{k}$-vector spaces. An element $v$ is homogeneous if $v\in V_k$ for some $k$, and its degree is $|v|=k$. A linear subspace of $V$ is a subset $V'\subseteq V$, such that $V'_k=V'\cap V_k$ is subspace of $V_k$ for all $k$. The direct sum and tensor product of two graded vector spaces $V$ and $W$ are also graded vector spaces:
\[
	(V\oplus W)_k=V_k\oplus W_k,\qquad\qquad (V\otimes W)_k=\oplus_{i+j=k} V_i\otimes W_j.
\]
Particular, $T(V)=\oplus_{n=0}^{\infty}(\otimes^n V)$ is a graded algebra with respect to the tensor product, called the tensor algebra of $V$. The symmetric algebra $S(V)$ of $V$ is the quotient of $T(V)$ by the ideal generated by $\{ v\otimes w - (-1)^{|v||w|}w\otimes v\, |\, v,w\in V \text{ homogeneous}\}$. The image of $v_1\otimes \cdots \otimes v_n$ under the natural projection $T(V)\to S(V)$ is denoted by $v_1\odot\cdots\odot v_n$.

\par For any $n\in\mathbb{Z}$, the $n$-th suspension of $V$, denoted by $V[n]$, is a graded vector space with grading  $(V[n])_k=V_{k+n}$. A linear map $f:V\to W$ between graded vector spaces is a collection of linear maps $\{f_k:V_k\to W_k\}_{k\in\mathbb{Z}}$. A linear map of degree $n$ from $V$ to $W$ is a linear map $g:V\to W[n]$. Such a $g$ is homogeneous, and its degree is denoted by $|g|$. For any $v\in V$, we use $v[n]$ to denote the corresponding element in $V[n]$.

\par A linear map $m_n:\otimes^n V\to V$ is graded symmetric if 
\[
	m_n(v_1\otimes \cdots \otimes v_{i+1}\otimes v_i \otimes \cdots \otimes v_n)=(-1)^{|v_i||v_{i+1}|}m_n(v_1\otimes \cdots \otimes v_i\otimes v_{i+1} \otimes \cdots \otimes v_n)
\]
for any $v_1,\cdots,v_n\in V$ homogeneous and $i=1,\cdots, n-1$. For a general $\tau$ in the symmetric group $S_n$ of $n$ letters, the Koszul sign $e(\tau)$ is introduced so that 
\[
	m_n(v_{\tau(1)}\otimes \cdots \otimes v_{\tau(n)})=e(\tau)m_n(v_1\otimes \cdots \otimes v_n).
\]
Be alerted that $e(\tau)$ also depends on $|v_1|,\cdots,|v_n|$. A linear map $l_n:\otimes^n V\to V$ is graded antisymmetric if 
\[
	l_n(v_{\tau(1)}\otimes \cdots \otimes v_{\tau(n)})=(-1)^\tau e(\tau)l_n(v_1\otimes \cdots \otimes v_n).
\]
Here, $(-1)^\tau$ is $1$ if $\tau$ is even, and $-1$ if $\tau$ is odd. 

\begin{definition}
\label{def:DGLA}
A \emph{graded Lie algebra} is a $\mathbb{Z}$-graded vector space $V$ equipped with a graded antisymmetric bilinear bracket $[\cdot,\cdot]:V\otimes V\to V$ of degree $0$ that satisfies the Jacobi identity, i.e.
\[
	[a,[b,c]]=[[a,b],c]+(-1)^{|a||b|}[b,[a,c]],\qquad \forall\, a,b,c\in V\text{ homogeneous}.
\]
A \emph{derivation} of degree $n$ is a linear map $d:V\to V[n]$ satisfying
\[
	d[a,b]=[da,b]+(-1)^{n|a|}[a,db].
\]
A derivation $d^+$ of degree $+1$ is called a \emph{positive differential} if $(d^+)^2=0$. Analogously, a derivation $d^-$ of degree $-1$ is called a \emph{negative differential} if $(d^-)^2=0$. Naturally, the triples $(V,[\cdot,\cdot],d^+)$ and $(V,[\cdot,\cdot],d^-)$ are called \emph{positive DGLA} and \emph{negative DGLA}, respectively. The quadruple $(V,[\cdot,\cdot],d^{+},d^-)$ is called a \emph{blended DGLA} if $d^+$ and $d^-$ satisfy the compatible condition
\[
	d^+d^-+d^-d^+=0.
\]
In this case, $d=d^++d^-$ is called a \emph{blended differential}.
\end{definition}
\begin{remark}
A positive differential (or a positive DGLA respectively) is just a differential (a DGLA) in the usual sense, and the adjective ``positive'' can be omitted without ambiguity.
\end{remark}
\par\noindent From Definition \ref{def:DGLA}, two derivations $d^+$ and $d^-$ make $(V,[\cdot,\cdot])$ into a blended DGLA if and only if they satisfy $|d^+|=1$, $|d^-|=-1$ and
\begin{eqnarray*}
	(d^++d^-)^2=0.
\end{eqnarray*}
\begin{definition}
An element $v$ of $(V,[\cdot,\cdot],d^+,d^-)$ is a \emph{Maurer-Cartan element} if $v$ is of the form $v_++v_-$ with $v_+\in V_1$ and $v_-\in V_{-1}$, and satisfies the \emph{Maurer-Cartan equation}
\begin{eqnarray*}
	dv+\dfrac{1}{2}[v,v]=0\qquad (d=d^++d^-).\label{eq:mc_bdgla}
\end{eqnarray*}
\end{definition}
\par\noindent By counting degrees, the Maurer-Cartan equation above decomposes into a system of equations:
\begin{eqnarray*}
\begin{cases}
d^+v_++\frac{1}{2}[v_+,v_+] = 0,\\
d^+v_-+d^-v_++[v_+,v_-] = 0,\\
d^-v_-+\frac{1}{2}[v_-,v_-] = 0.
\end{cases}
\end{eqnarray*}

\par Next, let us recall the definition of an $L_\infty$-algebra and extend it to a blended version. Let $S_{j,n-j}$ be the collection of $(j,n-j)$ shuffles in $S_n$, i.e. $\tau\in S_{j,n-j}$ if and only if $\tau(1)<\cdots<\tau(j)$ and $\tau(j+1)<\cdots<\tau(n)$.

\begin{definition}
\label{def:l_infty}
An $L_\infty$-algebra is a $\mathbb{Z}$-graded vector space $V$ together with a family of graded symmetric linear maps $\{m_k:\otimes^k V\to V[1]\}_{k\ge 0}$, such that the family of \emph{Jacobiators} $\{J_n:\otimes^n V\to V[2]\}_{n\ge 1}$, defined by
\begin{align*}
	J_n(v_1,\cdots,v_n)=\sum_{i+j=n+1} \sum_{\tau \in S_{j,n-j}}e(\tau)m_i(m_j(v_{\tau(1)},\cdots,v_{\tau(j)}),v_{\tau(j+1)},\cdots,v_{\tau(n)}),\label{eq:jacobiator}
\end{align*}
vanishes identically. The $L_\infty$-algebra $V$ is \emph{flat} if $m_0=0$. A linear subspace $V'\subset V$ is an $L_\infty$-subalgebra if all $m_k$'s are closed on $V'$, i.e. $m_k(\otimes^k V')\subseteq V'[1]$ for all $k$. In this case, $(V',\{m_k|_{V'}\}_{k\ge 0})$ is itself an $L_\infty$-algebra.
\end{definition}

\begin{remark}
\label{re:l_infty1}
The notion of $L_\infty$-algebra in Definition \ref{def:l_infty} is traditionally referred to as an $L_\infty[1]$-algebra. Instead, an $L_\infty$-structure on $V$ is a collection of graded antisymmetric linear maps $\{l_k:\otimes^k V\to V[2-k]\}_{k\ge 0}$, satisfying
\begin{align*}
	\sum_{i+j=n+1} (-1)^{i(j-1)}\sum_{\tau \in S_{j,n-j}}e(\tau)l_i(l_j(v_{\tau(1)},\cdots,v_{\tau(j)}),v_{\tau(j+1)},\cdots,v_{\tau(n)})=0.
\end{align*}
A DGLA $(V,[\cdot,\cdot],d)$ is a flat $L_\infty$-algebra in this sense, by setting $l_1=d$, $l_2(\cdot,\cdot)=[\cdot,\cdot]$ and $l_k=0$ for all $k\ge 3$. An $L_\infty$-structure on $V$ is equivalent to an $L_\infty[1]$-structure on $V[1]$ by the ``\emph{shift-isomorphism}'' sh$:\otimes^k V[1] \to (\otimes^k V)[k]$
\[
    v_1[1]\otimes\cdots\otimes v_k[1] \mapsto (-1)^{\sum_{i=1}^k (k-i)|v_i|}v_1\otimes \cdots \otimes v_k.
\]
Readers may find more details in \cite{ls}. In the rest of the paper, we shall work with $L_\infty[1]$-algebras exclusively, and abuse the language to call them $L_\infty$-algebras.
\end{remark}

\par Since in Definition \ref{def:l_infty} the map $m_1$ functions as a differential that makes $V$ into a cochain complex, the $L_\infty$-algebra defined in this way is viewed as cohomological analogue. The homological analogue of $L_\infty$-algebra is defined in a similarly way, with the distinction that all structure maps $m_k$ are of degree $-1$ instead of $1$. To distinguish the two types of $L_\infty$-algebras, we shall call the $L_\infty$-algebra in Definition \ref{def:l_infty} a \emph{positive $L_\infty$-algebra}, and the homological analogue a \emph{negative $L_\infty$-algebra}. As before, the adjective ``positive'' can be omitted. A blended $L_\infty$-algebra is a compatible combination of the two analogues of $L_\infty$-algebras.

\begin{definition}
\label{def:blended_l_infty}
A blended $L_\infty$-algebra is a $\mathbb{Z}$-graded vector space $V$ equipped with two families of graded symmetric linear maps $\{m_k:\otimes^k V\to V[1]\}_{k\ge 0}$ and $\{n_k:\otimes^k V\to V[-1]\}_{k\ge 0}$, such that the Jacobiators of the family $\{m_k+n_k\}$ vanish.
\end{definition}
\par\noindent Equivalently, $(V,\{m_k\},\{n_k\})$ is a blended $L_\infty$-algebra if and only if $(V,\{m_k\})$ is a positive $L_\infty$-algebra, $(V,\{n_k\})$ is a negative $L_\infty$-algebra, and $\{m_k\}$ and $\{n_k\}$ are compatible in the sense that
\begin{eqnarray}
	\sum_{i+j=n+1} \sum_{\tau \in S_{j,n-j}}e(\tau)\; [m_i(n_j(v_{\tau(1)},\cdots,v_{\tau(j)}),v_{\tau(j+1)},\cdots,v_{\tau(n)}) & & \nonumber \\
    +n_i(m_j(v_{\tau(1)},\cdots,v_{\tau(j)}),v_{\tau(j+1)},\cdots,v_{\tau(n)}) ] & = & 0,
\end{eqnarray}
for any $n\in\mathbb{N}$ and $v_1,\cdots, v_n\in V$. It is known that a positive $L_\infty$-structure on $V$ is equivalent to a codifferential $Q$ of the graded coalgebra $(S(V),\Delta_S)$ (\cite{ls}), in which $\Delta_S$ is the coproduct defined by
\[
	\Delta_S (v_1\odot \cdots \odot v_n) = \sum_{i=0}^n (v_1\odot \cdots\odot v_i)\otimes (v_{i+1}\odot \cdots \odot v_n).
\]
One may easily extend this viewpoint to negative and blended $L_\infty$-algebras by changing $Q$ to a negative or blended codifferential.

\begin{definition}
A \emph{Maurer-Cartan element} of $(V,\{m_k\},\{n_k\})$ is an element $v$ of degree $0$ that satisfies the \emph{Maurer-Cartan equation}
\begin{align}
	\sum_{k=0}^{\infty}\frac{1}{k!} (m_k+n_k)(v,\cdots,v)=0.\label{eq:mc_blinfty}
\end{align}
\end{definition}
\par\noindent A blended $L_\infty$-algebra degenerates to a positive $L_\infty$-algebras if the family $\{n_k\}$ vanishes. Consequently, Maurer-Cartan elements of a positive $L_\infty$-algebra should satisfy
\[
	\sum_{k=0}^{\infty}\frac{1}{k!} m_k(v,\cdots,v)=0.
\]
To work with Maurer-Cartan elements, convergence of the Maurer-Cartan equation needs to be ensured. We shall discuss the conditions required for convergence when we meet the problem later in this paper.

\par Higher derived brackets are invented to construct $L_\infty$-algebras in \cite{v1}. The idea also works for blended $L_\infty$-algebra, and we shall state the construction in this setup.
\begin{definition}
A V-algebra is a graded Lie algebra $(V,[\cdot,\cdot])$ that is the direct sum of an abelian Lie subalgebra $\mathfrak{a}$ and a Lie subalgebra $\mathfrak{b}$.
\end{definition}
\par\noindent Denote the projection $V\to \mathfrak{a}$ by $P$, then $\mathfrak{b}=\text{ker}(P)$. The condition that $\mathfrak{b}$ is a Lie subalgebra is equivalent to
\[
	P[x,y]=P[Px,y]+P[x,Py],\quad \forall x,y\in V.
\]
Therefore, a V-algebra is characterized by the quadruple $(V,[\cdot,\cdot],\mathfrak{a},P)$. An element $\Delta\in V$ is called a \emph{positive (negative, or blended, respectively) Maurer-Cartan element} if $[\Delta,\cdot]=0$ is a positive (negative, or blended) differential of $(V,[\cdot,\cdot])$.

\begin{theorem}[\cite{v1}]
\label{thm:vo}
Let $(V,[\cdot,\cdot],\mathfrak{a},P)$ be a V-algebra. For any blended Maurer-Cartan element $\Delta\in V$, the associated family of \emph{higher derived brackets} $\{m^k_\Delta: \otimes^k\mathfrak{a}\to \mathfrak{a}\}_{k\ge 1}$, defined by
\begin{eqnarray*}
	&&m_\Delta^0=P(\Delta),\\
	&&m_\Delta^k (a_1,\cdots,a_k)=P[\cdots[[\Delta,a_1],a_2],\cdots,a_k],
\end{eqnarray*}
makes $\mathfrak{a}$ into a blended $L_\infty$-algebra, denoted by $\mathfrak{a}_\Delta^P$. Moreover, $\mathfrak{a}_\Delta^P$ is flat if $P(\Delta)=0$.
\end{theorem}
\par\noindent The theorem can be proved by showing that the Jacobiators of the family $\{m_k^\Delta\}_{k\ge 1}$ coincide with the higher derived brackets associated with $\frac{1}{2}[\Delta,\Delta]$.

\par If $\mathfrak{a}_{\Delta}^P$ is positive, there is also an $L_\infty$-structure on $V[1]\oplus \mathfrak{a}$ (\cite{v2}) that can be used to control simultaneous deformations (\cite{cs}). Due to the mismatch of the degrees in the DGLA $(V,[\cdot,\cdot],[\Delta,\cdot])$ and in the $L_\infty$-algebra $\mathfrak{a}_\Delta^P$, this construction does not generalize to blended $L_\infty$-algebras constructed via higher derived brackets.
\begin{theorem}[\cite{v2,cs}]
\label{thm:cs}
Given a positive Maurer-Cartan element $\Delta$ of the V-algebra $(V,[\cdot,\cdot],\mathfrak{a},P)$, the direct sum $V[1]\oplus\mathfrak{a}$ is an $L_\infty$-algebra, denoted by $(V[1]\oplus\mathfrak{a})_\Delta^P$, whose structure maps are given by
\begin{align}
	& m_0=P(\Delta),\label{eq:comb_l_inf1}\\
	& m_1(v[1],a)=(-[\Delta,v][1],P(v+[\Delta,a])),\label{eq:comb_l_inf2}\\
	& m_2(v[1],w[1])=(-1)^{|v|}[v,w][1],\label{eq:comb_l_inf3}\\
	& m_k(v[1],a_1,\cdots,a_{k-1})=P[\cdots[[v,a_1],a_2],\cdots,a_{k-1}],\quad k\ge 2,\label{eq:comb_l_inf4}\\
	& m_k(a_1,\cdots,a_k)=P[\cdots[[\Delta,a_1],a_2],\cdots,a_k],\quad k\ge 2,\label{eq:comb_l_inf5}
\end{align}
for any $a,a_1,\cdots,a_k\in\mathfrak{a}$ and $v,w\in V$ homogeneous, and the structure maps at all other combinations are zero. The $L_\infty$-algebra $(V[1]\oplus\mathfrak{a})_\Delta^P$ is flat if $\Delta\in\text{ker}(P)$. Furthermore, $(\tilde{\Delta},a)\in V_1\oplus\mathfrak{a}_0$ is a Maurer-Cartan element if and only if
\[
	\begin{cases}
    	[\Delta+\tilde{\Delta},\Delta+\tilde{\Delta}]=0, \text{ and}\\
        a\text{ is a Maurer-Cartan element of }\mathfrak{a}_{\Delta+\tilde{\Delta}}^P.
	\end{cases}
\]
\end{theorem}

\section{Graded Q-manifolds and Submanifolds}
\par The regular homological vector field in literature is not adequate for our applications. In this section, we shall extend the concept of homological vector field to include multi-vector fields, and generalize it to a blended version that allows a component of degree $-1$. The advantage of doing this is that by this viewpoint Poisson manifolds, Lie algebroids and Courant algebroids are unified elegantly as blended $Q$-manifolds (Example \ref{eg:poisson}, Example \ref{eg:lie_algd}, Lemma \ref{lem:courant_q}). The submanifold of coisotropic type in a blended $Q$-manifold is defined, which turns out to encompass the structures including coisotropic submanifolds (Example \ref{eg:lie_algd_sub}), Lie subalgebroids (Example \ref{eg:poisson_mfd_sub}), and Dirac structures (Lemma \ref{lem:dirac_q}). 

\par Essentially all graded manifolds are from graded vector bundles (\cite{b}). Therefore, we adopt the following definition.

\begin{definition}
A graded manifold is a $\mathbb{Z}$-graded vector bundle $\mathcal{E}$, i.e. a collection of ordinary vector bundles $\{\mathcal{E}_k\}_{k\in\mathbb{Z}}$ of finite rank over some manifold $M$.
\end{definition}
\par\noindent The elements in $\mathcal{E}_k$ are homogeneous of degree $k$. For any $n\in\mathbb{Z}$, the $n$th-suspension $\mathcal{E}[n]$ is a graded vector bundle with $(\mathcal{E}[n])_k=\mathcal{E}_{n+k}$. The dual $\mathcal{E}^*$ of $\mathcal{E}$ is also a graded vector bundle with grading $(\mathcal{E}^*)_k=(\mathcal{E}_{-k})^*$. Moreover, any regular vector bundle $A\to M$ can be considered as a graded vector bundle concentrated in degree $0$, still denoted by $A$.

\par\noindent In this paper, we are only interested in graded manifold of the form $\{E_k\}_{k=-n}^{-1}$ $(n\in\mathbb{Z}^+)$, i.e. it only has finitely many nonzero components and all of them are concentrated in degrees between $-n$ and $-1$ with $\mathcal{E}_{-n}\ne 0$. We call $n$ the degree of the graded manifold.

\par Fixing $x\in M$, the fiber $\mathcal{E}_x$ is a graded vector space whose symmetric algebra $S(\mathcal{E}_x)$ is defined. Then $S(\mathcal{E})=\bigsqcup_{x\in M} S(\mathcal{E}_x)$ becomes a graded vector bundle, called the \emph{symmetric algebra bundle} of $\mathcal{E}$. The algebra of smooth functions on $\mathcal{E}$, denoted by $\mathcal{C}^\infty(\mathcal{E})$, is the unital graded commutative associative algebra $\Gamma(S(\mathcal{E}^*)))$.
\begin{example}
Given a smooth manifold $M$, the graded vector bundle $\mathcal{E}=T[1]M$ is obtained by shifting the degree of the fibers of $TM$ by $1$. It follows that $\mathcal{C}^\infty(\mathcal{E})$ coincides with the algebra $\Omega(M)$ of differential forms on $M$.  
\end{example}

\par The derivation of $\mathcal{C}^\infty(\mathcal{E})$ are vector fields.
\begin{definition}
A \emph{vector field} $X$ of degree $|X|=k$ on $\mathcal{E}$ is a linear map $\mathcal{C}^\infty(\mathcal{E})\to \mathcal{C}^\infty(\mathcal{E})[k]$ that satisfies the Leibniz rule, i.e.
\[
	X(fg)=X(f)g+(-1)^{k|f|}fX(g)
\]
for any $f,g\in\mathcal{C}^\infty(\mathcal{E})$ homogeneous. The collection of vector fields on $\mathcal{E}$ is a graded vector space, denoted by $\mathcal{V}(\mathcal{E})$.
\end{definition}
\par\noindent Let $(x^1,\cdots,x^m)$ be a local coordinate system on an open subset $U\subset M$, and $\{\xi_k^j\}$ the linear coordinates of the fibers of $E_k|_U$. Locally, a vector field has the form
\[
	f^i\dfrac{\partial}{\partial x^i}+g^j_k \dfrac{\partial}{\partial \xi_k^j},\qquad f^i,g^j_k\in\mathcal{C}^\infty(\mathcal{E}).
\]
Here $\dfrac{\partial}{\partial x^i}$ and $\dfrac{\partial}{\partial \xi^j_k}$ have degrees inherited from the grading of $\mathcal{E}$, defined by $\bigg|\dfrac{\partial}{\partial x^i} \bigg| =0$, and $\bigg|\dfrac{\partial}{\partial \xi^j_k} \bigg|=-|\xi^j_k|$. As a locally free $\mathcal{C}^\infty(\mathcal{E})$-module, the set $\mathcal{V}(\mathcal{E})$ is in fact the section space of a graded vector bundle $T\mathcal{E}$ over $\mathcal{E}$, called the tangent bundle. For each $e\in\mathcal{E}$, the tangent space  $T_e\mathcal{E}$ is a graded vector space. Denote by $T_e[-1]\mathcal{E}$ the degree shifted by $-1$ in the fibers. It follows that the collection $\sqcup_{e\in \mathcal{E}} S(T_e[-1]\mathcal{E})$ forms a graded vector bundle over $\mathcal{E}$, the sections of which are \emph{multi-vector fields}. Locally, a multi-vector field is a finite sum of terms of the form
\[	
	f^{i_1\cdots i_p j_1\cdots j_q}_{k_1\cdots k_q} \frac{\partial}{\partial x^{i_1}}\wedge\cdots\wedge \frac{\partial}{\partial x^{i_p}}\wedge\frac{\partial}{\partial \xi_{k_1}^{j_1}}\wedge\cdots\wedge\frac{\partial}{\partial \xi_{k_q}^{j_q}}.
\]
Here, the operation $v\wedge w$ of vector fields is understood as $(v[-1])\odot (w[-1])$. Beyond the degree inherited from $\mathcal{E}$, a multi-vector field is also graded by its multiplicity. Therefore, a homogeneous multi-vector field $Z$ of multiplicity $l$ has bi-grading $(|Z|,l)$. In fact, the grading on $\mathfrak{X}(\mathcal{E})$ is such selected so that the degree of $Z$ is the sum $|Z|+l$. In the following, we shall call this sum the \emph{total degree} of $Z$, and denote it by $|Z|_\text{total}$.

\par\noindent A multi-vector field of multiplicity $l$ is called an $l$-vector field. Denote by $\mathfrak{X}^l(\mathcal{E})$ the collection of $l$-vector fields on $\mathcal{E}$. It is clear that $\mathfrak{X}^0(\mathcal{E})$ is just $\mathcal{C}^\infty(\mathcal{E})$, and $\mathfrak{X}^1(\mathcal{E})$ coincides with $\mathcal{V}(\mathcal{E})$. For a homogeneous function $f$ and a homogeneous vector field $X$, one has $|f|_\text{total}=|f|$ and $|X|_\text{total}=|X|+1$. The space $\mathfrak{X}(\mathcal{E})=\sum_{l=0}^\infty\mathfrak{X}^l(\mathcal{E})$ is the collection of all multi-vector fields, which is a $\mathcal{C}^\infty(\mathcal{E})$-module as well as a graded commutative associative algebra under the product $\wedge$. Precisely, we have
\[
	Y\wedge Z=(-1)^{|Y|_\text{total}\cdot |Z|_\text{total}}Z\wedge Y,\qquad \forall Y,Z\in\mathfrak{X}(\mathcal{E})\text{ homogeneous}.
\]

\par\noindent For a smooth manifold $M$, it is known that $\mathfrak{X}(M)$ is a Gerstenhaber algebra. This picture can be generalized to $\mathfrak{X}(\mathcal{E})$.

\begin{theorem}
With respect to the total degree of multi-vector fields, $\mathfrak{X}(\mathcal{E})$ together with the product $\wedge$ and the Schouten-Nijenhuis bracket $[\cdot,\cdot]_\text{SN}$ of multi-vector fields forms a Gerstenhaber algebra. Here, the Schouten-Nijenhuis bracket is determined by
\par (1) $[f,g]_\text{SN}=0$,
\par (2) $[X,f]_\text{SN}=X(f)$,
\par (3) $[X,Y]_\text{SN}=X\circ Y-(-1)^{|X||Y|}Y\circ X$,
\par (4) $[Z_1,Z_2]=-(-1)^{(|Z_1|_\text{total}-1)(|Z_2|_\text{total}-1)}[Z_2,Z_1]_\text{SN}$,
\par (5) $[Z,Z_1\wedge Z_2]_\text{SN}=[Z,Z_1]_\text{SN}\wedge Z_2+(-1)^{(|Z|_\text{total}-1)(|Z_1|_\text{total})}Z_1\wedge[Z,Z_2]_\text{SN}$,\\
for any $f,g\in\mathcal{C}^\infty(\mathcal{E})$, $X,Y\in \mathfrak{X}^1(\mathcal{E})$, and $Z,Z_1,Z_2\in\mathfrak{X}(\mathcal{E})$ homogeneous.
\end{theorem}
\par\noindent We only need to prove $(\mathfrak{X}(\mathcal{E}[1],[\cdot,\cdot]_\text{SN})$ is a graded Lie algebra by showing that it satisfies the Jacobi identity. This can be verified by straightforward but tedious induction on the multiplicity of multi-vector fields. We ask the reader to excuse us for omitting the verification here.

\par\noindent Since we shall work with the graded Lie algebra structure on $\mathfrak{X}(\mathcal{E})[1]$ exclusively, for simplicity, denote the degree of $\mathfrak{X}(\mathcal{E})[1]$ by $\|\cdot\|$, i.e. $\|Z\|=|Z|_\text{total}-1$. Without further notice, we shall use this shifted total degree $\|\cdot\|$ for the graded Lie algebra $(\mathfrak{X}(\mathcal{E}),[\cdot,\cdot]_\text{SN})$. It is worth noticing that the degrees $\|\cdot\|$ and $|\cdot|$ coincide on $\mathfrak{X}^1(\mathcal{E})$. Furthermore, $\mathfrak{X}^1(\mathcal{E})$ is close under $[\cdot,\cdot]_{SN}$, and we have the following easy fact. 
\begin{corollary}
$\mathfrak{X}^1(\mathcal{E})$ is a graded lie subalgebra of $(\mathfrak{X}(\mathcal{E}),[\cdot,\cdot]_\text{SN})$.
\end{corollary}

\par Next, let us define and generalize homological vector fields.

\begin{definition}
A multi-vector field $Q^+\in\mathfrak{X}(\mathcal{E})$ is called \emph{homological} if $\|Q^+\|=1$ and $[Q^+,Q^+]_{\text{SN}}=0$. The pair $(\mathcal{E},Q^+)$ is called a graded $Q$-manifold or a positive $Q$-manifold.
\par A \emph{negative homological vector field} is a multi-vector field $Q_-$ of total degree $-1$ satisfying $[Q_-,Q_-]_\text{SN}=0$. A \emph{blended $Q$-structure} on $\mathcal{E}$ is a pair consisting of a positive homological vector field $Q_+$ and a negative homological vector field $Q_-$ that are compatible by satisfying
\[
	Q_+Q_-+Q_-Q_+=0.
\]
In this case, $Q=Q_++Q_-$ is called a \emph{blended homological vector field}.
\end{definition}

\par\noindent Concerning the differential attribute of homological vector fields, the The following observation is immediate.
\begin{proposition}
A homological vector field $Q$ gives rise to a differential $[Q,\cdot]$ of $(\mathfrak{X}(\mathcal{E}),[\cdot,\cdot]_\text{SN})$ and makes it into a differential graded Lie algebra. If $Q$ is negative or blended, one gets negative or blended differential graded Lie algebra structures respectively.
\par If $Q$ is a vector field, i.e. the multiplicity of $Q$ is $1$, the respective differential graded Lie algebra structure can be restricted on $\mathfrak{X}^1(\mathcal{E})$ to make it a differential graded Lie algebra of a version the same as $Q$. 
\end{proposition}

\begin{example}
\label{eg:poisson}
Any smooth manifold $M$ can be considered as a degenerated graded vector bundle without fibers. A Poisson structure on $M$ is a bi-vector field $\pi$ satisfying $[\pi,\pi]_\text{SN}=0$. Since $\|\pi\|=1$, we conclude that a Poisson manifold $(M,\pi)$ is equivalent to a (positive) $Q$-manifold of degree $0$.
\par In general, a blended $Q$-manifold of degree $0$ is a smooth manifold together with $Q^+=\pi\in\mathfrak{X}^2(M)$ and $Q^-=f\in\mathcal{C}^\infty(M)$ satisfying
\[
	[\pi,\pi]_\text{SN}=0,\quad \text{ and }\quad [\pi,f]_\text{SN}=0.
\]
This is equivalent to that $(M,\pi)$ is a Poisson manifold, and $f$ is a Casimir function.
\end{example}

\begin{example}
\label{eg:lie_algd}
Let $(A,[\cdot,\cdot],\rho)$ be a Lie algebroid over $M$. It is well-known that the Lie algebroid structure on $A$ is equivalent to its Lie algebroid differential $d_A:\Gamma(\wedge^n A^*)\to \Gamma(\wedge^{n+1}A^*)$ defined by
\begin{eqnarray*}
	d_A \omega(a_0,\cdots,a_n) &=&\sum_{i=0}^n \rho(a_i)(\omega(a_0,\cdots,\hat{a}_i,\cdots,a_n))\\
    &+& \sum_{0\le i<j\le n} (-1)^{i+j}\omega([a_i,a_j],a_0,\cdots,\hat{a}_i,\cdots,\hat{a}_j,\cdots,a_n),
\end{eqnarray*}
which satisfies $d_A^2=0$. Let us consider the graded vector bundle $A[1]$. It is easy to see $\mathcal{C}^\infty(A[1])=\Gamma(\wedge^\cdot A^*)$. It follows that $d_A$ can be viewed as a vector field of degree $1$ on $A[1]$ that satisfies $[d_A,d_A]_\text{SN}=0$. Therefore, Lie algebroids are graded $Q$-manifolds of degree $1$. 
\end{example}


\par A submanifold $\mathcal{A}$ of a graded manifold $\mathcal{E}$ is a graded vector bundle $\{\mathcal{A}_k\}_{k\in\mathbb{Z}}$ over a submanifold $S$ of $M$ such that $\mathcal{A}_k$ is a subbundle of $\mathcal{E}_k$. With the homological vector field $Q$ in consideration, we propose the following definition. 
\begin{definition}
\label{def:coiso}
A submanifold $\mathcal{A}$ of $(\mathcal{E},Q)$ is called \emph{of coisotropic type} if $\mathcal{A}$ is concentrated in some fixed degree $k_0$, and $Q|_{\mathcal{E}_{k_0}} (N^\bot \mathcal{A})=0$. Here,
\[
	N^\bot \mathcal{A}=\sqcup_{s\in S}\{ \xi \in T^*_s\mathcal{E}_{k_0}\,|\, \xi(T_s\mathcal{A})=0 \}
\]
is the conormal bundle of $\mathcal{A}$ in $\mathcal{E}_{k_0}$.
\end{definition}

\begin{example}
\label{eg:poisson_mfd_sub}
When $(M,\pi)$ is Poisson, a submanifold $S$ of coisotropic type is just a coisotropic submanifold in the usual sense. This is also the reason we select this name for the submanifold in Definition \ref{def:coiso}.
\par\noindent For a general blended $Q$-manifold $M$ of degree $0$, we have $Q=\pi+f$ with $\pi$ a Poisson bi-vector field and $f$ a Casimir function. A submanifold $S$ is of coisotropic type if and only if $S$ is a regular coisotropic submanifold of the Poisson manifold $(M,\pi)$, and satisfies $f|_S=0$.  
\end{example}

\begin{example}
\label{eg:lie_algd_sub}
Given a Lie algebroid $(A,[\cdot,\cdot],\rho)$, a subbundle $E$ over $S\subset M$ is a Lie subalgebroid if and only if the homological vector field $Q=d_A$ is tangent to $E[1]$ (\cite{j}). Equivalently, $Q|_{E[1]}$ annihilates the conormal bundle $N^\bot E[1]$. Thus the submanifolds of coisotropic type in $(A[1],d_A)$ are in 1-1 correspondence with the Lie subalgeboids of $A$. 
\end{example}

\section{Deformations in Graded $Q$-manifolds}
\label{sec:deform}
\par We shall consider the deformations of a blended $Q$-manifold, and of a submanifold of coisotropic type in this section.

\subsection{Deformation of $Q$-structure}
\par Let $(\mathcal{E},Q)$ be a blended $Q$-manifold. A deformation of the $Q$-structure is an element $\tilde{Q}$ of the form $\tilde{Q}^++\tilde{Q}^-\in \mathfrak{X}_1(\mathcal{E})\oplus \mathfrak{X}_{-1}(\mathcal{E})$. Here we use $\mathfrak{X}_k(\mathcal{E})$ to denote the collection of multi-vector field $Z$ whose shifted degree $\|Z\|=k$. It is easy to see $Q+\tilde{Q}$ is homological if and only if $\tilde{Q}$ satisfies the Maurer-Cartan equation
\[
	[Q,\tilde{Q}]_\text{SN}+\frac{1}{2}[\tilde{Q},\tilde{Q}]_\text{SN}=0.
\]
\begin{theorem}
\label{thm:q_deform}
Deformations of a blended $Q$-manifold $(\mathcal{E},Q)$ is controlled by the blended DGLA $(\mathfrak{X}(\mathcal{E}),[\cdot,\cdot]_\text{SN},[Q,\cdot]_\text{SN})$. To be specific, $\tilde{Q}=\tilde{Q}^++\tilde{Q}^-\in \mathfrak{X}_1(\mathcal{E})\oplus \mathfrak{X}_{-1}(\mathcal{E})$ is a deformation of $Q$ if and only if $\tilde{Q}$ is a Maurer-Cartan element.
\end{theorem}
\par\noindent Let $Q^+$ and $Q^-$ be the degree $+1$ and degree $-1$ component of $Q$ respectively. Decomposing the components by degree, the above Maurer-Cartan equation is equivalent to the system
\begin{eqnarray*}
\begin{cases}
	[Q^+,\tilde{Q}^+]_\text{SN}+\dfrac{1}{2}[\tilde{Q}^+,\tilde{Q}^+]_\text{SN}=0\\
    [Q^+,\tilde{Q}^-]_\text{SN}+[Q^-,\tilde{Q}^+]_\text{SN}+[\tilde{Q}^+,\tilde{Q}^-]_\text{SN}=0\\
    [Q^-,\tilde{Q}^-]_\text{SN}+\dfrac{1}{2}[\tilde{Q}^-,\tilde{Q}^-]_\text{SN}=0
\end{cases}.
\end{eqnarray*}
If $Q$ is of multiplicity $1$, we may restrict our attention to deformations of the same type, i.e. $\tilde{Q}$ is also of multiplicity $1$. It follows that such deformations are controlled by the Lie subalgebra $(\mathfrak{X}^1(\mathcal{E}),[\cdot,\cdot]_\text{SN},[Q,\cdot]_\text{SN})$.
\begin{corollary}
If the homological vector field $Q$ has multiplicity $1$, a vector field $\tilde{Q}\in\mathfrak{X}^1(\mathcal{E})$ is a deformation of $Q$ if and only if $\tilde{Q}$ is a Maurer-Cartan element of $(\mathfrak{X}^1(\mathcal{E}),[\cdot,\cdot]_\text{SN},[Q,\cdot]_\text{SN})$.
\end{corollary}
\par\noindent This corollary also includes the result on deformations of a classical homological vector field, by imposing the condition $|Q|=|\tilde{Q}|=1$.

\begin{example}
A Poisson manifold $(M,\pi)$ is a $Q$-manifold of degree $0$. Considering it as a blended $Q$-manifold, a deformation of $\pi$ is of the form $\tilde{\pi}+f$ with $\tilde{\pi}\in\mathfrak{X}^2(M)$ and $f\in\mathcal{C}^\infty(M)$ that satisfies
\[
	\begin{cases}
		X_f=[\pi,f]_\text{SN}=0,\\
        [\pi,\tilde{\pi}]_\text{SN}+\dfrac{1}{2}[\tilde{\pi},\tilde{\pi}]_\text{SN}=0.
	\end{cases}
\]
Here $X_f$ is the Hamiltonian vector field of $f$. If restricting attention to deformations with $f=0$, we recover the classical result that $\tilde{\pi}\in\mathfrak{X}^(M)$ is a deformation of $\pi$ if and only if it is a Mauer-Cartan element of the DGLA $(\mathfrak{X}^1(M),[\cdot,\cdot]_\text{SN},[\pi,\cdot]_\text{SN})$.
\end{example}

\begin{example}
A Lie algebroid $(A,[\cdot,\cdot],\rho)$ over $M$ is equivalent to the $Q$-manifold $(A[1],d_A)$ of degree $1$, in which the homological vector field $d_A$ is the Lie algebroid differential. As a blended $Q$-manifold, a deformation of $d_A$ is a pair $(\tilde{d}_A,f)$ with $\tilde{d}_A$ a derivation of the graded commutative algebra $\Gamma(\wedge^\cdot E^*)$ and $f$ a smooth function on $M$, satisfying
\[
\begin{cases}
	\rho^*(df)=0,\\
    [d_A, \tilde{d}_A]_\text{SN}+\dfrac{1}{2}[d_A,\tilde{d}_A]_\text{SN}=0.
\end{cases}
\]
If requiring that $f=0$, we recover that the deformations of a Lie algebroid structure in the usual sense is controlled by the DGLA $(\mathfrak{X}^1(A[1]),[\cdot,\cdot]_\text{SN},[d_A,\cdot]_\text{SN})$.
\end{example}

\subsection{$L_\infty$-algebra of a submanifold}
\par Let $\mathcal{A}$ over $S\subset M$ be a submanifold of $(\mathcal{E},Q)$ of coisotropic type. By definition, $\mathcal{A}$ has only one nontrivial component concentrated in degree $k_0$. Let $\mathcal{B}$ be a complement of $\mathcal{A}$, i.e. $\mathcal{E}_{k_0}|_S=\mathcal{A}\oplus \mathcal{B}$. Identify the normal bundle $NS$ as a tubular neighborhood of $S$. Denote the pull-backs of $\mathcal{A}$ and $\mathcal{B}$ along the bundle projection $NS\to S$ by $\tilde{\mathcal{A}}$ and $\tilde{\mathcal{B}}$ respectively.
\[
	\begin{CD}
		\tilde{\mathcal{A}} @>>>  \mathcal{A}\\
		@VVV	@VVV\\
		NS	@>>>  S,
	\end{CD}\qquad\qquad
    \begin{CD}
		\tilde{\mathcal{B}} @>>>  \mathcal{B}\\
		@VVV	@VVV\\
		NS	@>{p}>>  S.
	\end{CD}
\]
\begin{definition}
A \emph{deformation} of $\mathcal{A}$ is a pair $(\sigma,\psi)\in \Gamma(NS) \oplus \Gamma ((\tilde{\mathcal{A}})^*|_{S_\sigma}\otimes \tilde{\mathcal{B}}|_{S_\sigma})$ in which $S_\sigma$ is the graph of $\sigma$, so that the graph of $\psi$ is a submanifold of coisotropic type over $S_\sigma$.
\end{definition}
\par\noindent Shrinking the tubular neighborhood if necessary, we may assume $\mathcal{E}_{k_0}=\tilde{\mathcal{A}}\oplus \tilde{\mathcal{B}}$ on $NS$. Since we are only interested in deformations near $\mathcal{A}$, without loss of generality, assume $M=NS$. Secondly, $\psi$ can be identified with an element $\phi\in\Gamma(\mathcal{A}^*\otimes\mathcal{B})$. To do this, by the pull-back construction, we extend $\psi$ to an element $\tilde{\psi}\in\Gamma(\tilde{\mathcal{A}}^*\otimes \tilde{\mathcal{B}})$, then restrict $\tilde{\psi}$ to $S$.  Let $\mathfrak{a}'$ be the set $\Gamma(S(\mathcal{A}^*)\otimes (NS\oplus\mathcal{B}))$. Immediately, $\mathfrak{a}'_0=\Gamma(NS)\oplus\Gamma(\mathcal{A}^*\otimes \mathcal{B})$ consists of all possible deformations of $\mathcal{A}$.

\par Next, we construct an $L_\infty$-algebra associated with $\mathcal{A}$ via higher derived brackets, which is similar to the construction of the $L_\infty$-algebra associated with a coisotropic submanifold in \cite{cf} and that associated with a Lie subalgebroid in \cite{j}. Readers will see that the construction does not rely on $\mathcal{A}$ being of coisotropic type, and we shall consider $\mathcal{A}$ as a subbundle unless otherwise specified. Let $\mathfrak{a}$ be the collection $\Gamma\big(S(\mathcal{A}^*)\otimes \big( (\wedge^\cdot NS)[1] \oplus S(\mathcal{B}[-1])[1]\big) \big)$, and it will be shown that $\mathfrak{a}$ is the space where the $L_\infty$-structure lives. If denoting the pull-backs of $NS$ and $\mathcal{B}$ along $p:\mathcal{A}\to S$ by $p^!NS$ and $p^!\mathcal{B}$, respectively, we have $\mathfrak{a}=\Gamma(S(p^!NS[-1]\oplus p^!\mathcal{B}[-1])[1])$.
\[
	\begin{CD}
		p^!NS @>>>  NS\\
		@VVV	@VVV\\
		\mathcal{A}	@>{p}>>  S,
	\end{CD}\qquad\qquad
    \begin{CD}
		p^!\mathcal{B} @>>>  \mathcal{B}\\
		@VVV	@VVV\\
		\mathcal{A}	@>{p}>>  S.
	\end{CD}
\]
An important observation is that the elements of $\mathfrak{a}$ can be viewed as multi-vector fields on $\mathcal{E}$, i.e, there is an injection $I:\mathfrak{a}\to \mathfrak{X}(\mathcal{E})$. For better understanding, we shall define $I$ locally first, and then give an intrinsic definition to show it does depend on the choice of local coordinates. Let $(x^i)$ be a local chart on an open subset $U\subset S$, $(y^j)$ the linear coordinates of the fibers of $NS|_U$. Let $(v_p)$ and $(w_q)$ be frames of $\mathcal{A}|_U$ and $\mathcal{B}|_U$, and  $(\xi^p)$, $(\eta^q)$ the dual linear coordinates of the fibers of $\mathcal{A}|_U$ and $\mathcal{B}|_U$ respectively. By the pull-back construction, $(\xi^p)$ and $(\eta^q)$ extend to linear coordinates of the fibers of $p^!\mathcal{A}$ and $p^!\mathcal{B}$, still denoted by $(\xi^p)$ and $(\eta^q)$. Finally, let $\gamma_k^{r_k}$ be the linear coordinates of the fibers of $\mathcal{E}_k$ for $k\ne k_0$. Then $I$ is a morphism of graded algebras determined by
\begin{eqnarray}
	F(x^i,\xi^p)\mapsto F(x^i,\xi^p),\qquad \sigma \mapsto \sigma^j\dfrac{\partial}{\partial y^j}, \qquad w_q\mapsto \dfrac{\partial}{\partial \eta^q}.\label{eq:I}
\end{eqnarray}
Here $(\sigma^j)$ are the components of $\sigma\in\Gamma(NS)$, and the image of $F\in\mathcal{C}^\infty(\mathcal{A})$ under $I$ is viewed as a function on $\mathcal{E}$ that is constant along the fibers of $p^!\mathcal{B}$, $p^!NS$ and $\mathcal{E}_{k}$ $(k\ne k_0)$. To define $I$ conceptually, consider the diagram
\[
	\begin{CD}
		T\mathcal{E}_{k_0} @>{Tp_2}>>  T(p^!\mathcal{A}) @>{Tp_1}>> T\mathcal{A}\\
		@VVV	@VVV	@VVV\\
		\mathcal{E}_{k_0}    @>{p_2}>>  p^!\mathcal{\mathcal{A}} @>{p_1}>> \mathcal{A},
	\end{CD}
\]
in which $Tp_i$ is the tangent map of the natural projection $p_i$ ($i=1,2$). The kernel of the composition $T(p_1\circ p_2)$ is just the pull-back of $NS\oplus \mathcal{B}$ along $\mathcal{A}\to S$, i.e. $p^!(NS\oplus\mathcal{B})$. The map $I$ turns out to be the composition
\[
	I: \mathfrak{a} \to \mathfrak{X}(\mathcal{E}_{k_0}) \to \mathfrak{X}(\mathcal{E}),
\]
which is easily seen an injection. Be altered that the degree adopted on $\mathfrak{a}$ is such selected to be compatible with the degree $\|\cdot\|=|\cdot|_\text{total}-1$ on $\mathfrak{X}(\mathcal{E})$.

\par\noindent A second observation is that $p!(NS \oplus \mathcal{B})$ is a tubular neighborhood of $\mathcal{A}$ in $\mathcal{E}_{k_0}$. In fact, we have $\mathcal{E}_{k_0}\cong p!(NS \oplus \mathcal{B})$ and shall identify them accordingly. Since for any vector bundle $F\to M$ one has $TF|_M=TM\oplus F$, substituting $F\to M$ by $\mathcal{E}_{k_0}=p!(NS \oplus \mathcal{B})\to \mathcal{A}$, we get a projection $T\mathcal{E}_{k_0}\to T\mathcal{E}_{k_0}|_\mathcal{A} \to p^!(NS\oplus \mathcal{B})$ and hence the map 
\begin{align*}
P:\mathfrak{X}(\mathcal{E})\to \mathfrak{X}(\mathcal{E}_{k_0})\to \mathfrak{a}.	
\end{align*}
Locally, the map $P$ is determined by
\begin{eqnarray*}
    \frac{\partial}{\partial x^i}\mapsto 0, \qquad \frac{\partial}{\partial y^j}\mapsto \frac{\partial}{\partial y^j}, \qquad \frac{\partial}{\partial \xi^p}\mapsto 0,\qquad  \frac{\partial}{\partial \eta^q}\mapsto \frac{\partial}{\partial \eta^q},\qquad \frac{\partial}{\partial \gamma_k^{r_k}}\mapsto 0,\\ 
    \eta^q\mapsto 0,\qquad \xi^p\mapsto \xi^p,\qquad \gamma_k^{r_k}\mapsto 0,\qquad f(x,y)\mapsto f(x,0) \quad \forall f\in\mathcal{C}^\infty(NS).
\end{eqnarray*}
\par\noindent It is straightforward to verify that $P\circ I=\text{Id}_{\mathfrak{a}}$ (using local coordinates for example). This indicates $\mathfrak{X}(\mathcal{E})$ splits as $\mathfrak{a}\oplus \text{ker}(P)$.
\begin{lemma}
$(\mathfrak{X}(\mathcal{E}),[\cdot,\cdot]_\text{SN},\mathfrak{a},P)$ is a V-algebra.
\end{lemma}
\begin{proof}
We only need to show that ker$(P)$ is a Lie algebra and $\mathfrak{a}$ is an abelian Lie subalgebra. Given an element $a\in\mathfrak{a}$, locally, its image $I(a)$ is a product of $F\in\mathcal{C}^\infty(\mathcal{A})$ and one or several terms in $\left\{\dfrac{\partial}{\partial y^j}, \dfrac{\partial}{\partial \eta^q}\right\}$. Since $F$ does not involve coordinates $(y^j)$ and $(\eta^q)$, one has $\dfrac{\partial}{\partial y^j}F=\dfrac{\partial}{\partial \eta^q}F=0$. This implies $[I(a_1),I(a_2)]_\text{SN}=0$ for any $a_1,a_2\in\mathfrak{a}$.
\par An element in $b\in\text{ker}(P)$ satisfies $P(b)=0$. Locally, this requires each summand of $b$ containing at least one factor from the collection $\mathsf{C}=\{f(x,y)\in\mathcal{C}^\infty(NS)| f(x,0)=0\}\cup\{\eta^q, \gamma_k^{r_k}\}\cup\left\{\dfrac{\partial}{\partial x^i}, \dfrac{\partial}{\partial \xi^p},\dfrac{\partial}{\partial \gamma_k^{r_k}}\right\}$. By exhausting all possible combinations of $b_1,b_2\in\text{ker}(P)$ in the form of monomials with a factor in $\mathsf{C}$, one gets $P[b_1,b_2]_\text{SN}=0$. So ker$(P)$ is a Lie subalgebra.
\end{proof}
\par\noindent Since the blended homological vector field $Q$ is a Maurer-Cartan element of the V-algebra $(\mathfrak{X}(\mathcal{E}),[\cdot,\cdot]_\text{SN},\mathfrak{a},P)$, we can apply Theorem \ref{thm:vo}.
\begin{theorem}
\label{thm:l_infty_of_coiso}
Fixing the choice of a tubular neighborhood of $S$ and a complement $\mathcal{B}$ of $\mathcal{A}$ in $\mathcal{E}_{k_0}$, the subbundle $\mathcal{A}$ is associated with an $L_\infty$-algebra $\mathfrak{a}_Q^P$ whose structure maps are given by the higher derived brackets of $Q$ below
\begin{eqnarray*}
	m_0 &=& P(Q)\\
    m_{k}(a_1,\cdots,a_k) &=& P[\cdots [Q,I(a_1)]_\text{SN},\cdots,I(a_k)],\qquad \forall a_1,\cdots, a_k\in \mathfrak{a}.
\end{eqnarray*}
Moreover, if $\mathcal{A}$ is of coisotropic type, the $L_\infty$-algebra $\mathfrak{a}_Q^P$ is flat.
\end{theorem}
\begin{proof}
Only the second part of the theorem needs proved. If $\mathcal{A}$ is of coisotropic type, one has $Q|_\mathcal{A}(N^\bot \mathcal{A})=0$. Since $\{dy^j,d\eta^q\}$ is a local frame of $N^\bot \mathcal{A}$, this forces $Q|_\mathcal{A}$ to have at least one factor from $\dfrac{\partial}{\partial x^i}$ and $\dfrac{\partial}{\partial \xi^p}$, but this is equivalent to $m_0=P(Q)=0$, i.e. $\mathfrak{a}_Q^P$ is flat.
\end{proof}

\par Notice that $\mathfrak{a}'=\Gamma(S(\mathcal{A}^*)\otimes (NS\oplus \mathcal{B}))$ is a subspace of $\mathfrak{a}$. If the higher derived brackets in Theorem \ref{thm:l_infty_of_coiso} are closed on $\mathfrak{a}'$, it becomes an $L_\infty$-subalgebra of $\mathfrak{a}_Q^P$. Since the image of $I':\mathfrak{a'}\overset{\tilde{I}}{\hookrightarrow} \mathfrak{a} \xrightarrow{I} \mathfrak{X}(\mathcal{E})$ consists of single vector fields only, the higher derived brackets of $Q$ are closed on $\mathfrak{a}'$ only when $Q$ is of multiplicity $1$. If this is the case, the $L_\infty$-structure on $\mathfrak{a}'$ can be constructed directly from a V-algebra. Consider the injection $I':\mathfrak{a}'\to \mathfrak{X}^1(\mathcal{E})$ and the projection
\begin{eqnarray*}
	P': \mathfrak{X}^1(\mathcal{E}) \xrightarrow{P|_{\mathfrak{X}^1(\mathcal{E})}} \mathfrak{a} \xrightarrow{\tilde{P}} \mathfrak{a'},
\end{eqnarray*}
where $\tilde{P}$ is the natural projection onto its direct summand. One has $\tilde{P}\circ \tilde{I}=\text{Id}_{\mathfrak{a}'}$, from which it is easy to derive $P'\circ I'=\text{Id}_{\mathfrak{a}'}$. Next, let us prove that $(\mathfrak{X}(\mathcal{E}), [\cdot,\cdot]_\text{SN},\mathfrak{a}', P')$ is a $V$-algebra: 
\par 1) it is straightforward that $I'(\mathfrak{a}')\subset I(\mathfrak{a})$ must be abelian; 
\par 2) ker$(P')$ is a $\mathcal{C}^\infty(\mathcal{E})$-module locally generated by the collection
\[
\left\{\dfrac{\partial}{\partial x^i}, \dfrac{\partial}{\partial \gamma^{r_k}_k}, \dfrac{\partial}{\partial \xi^p}, f\dfrac{\partial}{\partial y^j}, g\dfrac{\partial}{\partial \eta^q} \; \Big|\; f,g\in\mathcal{C}^\infty(\mathcal{E}), f|_\mathcal{A}=g|_\mathcal{A}=0\right\}.
\]
It is ready to verify that ker$(P)$ is a Lie subalgebra. Now we can apply Theorem \ref{thm:vo} again.
\begin{theorem}
\label{thm:l_infty_of_coiso2}
If $Q$ is of multiplicity $1$, then $(\mathfrak{X}^1(\mathcal{E}),[\cdot,\cdot]_\text{SN},\mathfrak{a}',P')$ is a V-algebra, and the higher derived brackets of $Q$ under the projection $P'$ induces an $L_\infty$-structure on $\mathfrak{a}'$, which is an $L_\infty$-subalgebra of $\mathfrak{a}_Q^P$.
\end{theorem}
\par\noindent This $L_\infty$-structure on $\mathfrak{a}'$ shall be denoted by ${\mathfrak{a}'}_Q^{P'}$.

\subsection{Deformation of submanifolds}

\par Given a possible deformation $a=(\sigma,\phi)$ of $\mathcal{A}$, i.e. $(\sigma,\phi)\in \mathfrak{a}'_0\subset\mathfrak{a}_0$, denote the vector field $I(a)=I(\sigma,\phi)$ by $X_{a}$, and the deformed subbundle determined by $(\sigma,\phi)$ by $\mathcal{A}_{\sigma,\phi}$. If $\sigma$ and $\phi$ have local coordinate expression $y^j=\sigma^j(x^i)$ and $\phi(v^p)=\phi^p_q w^q$, by the definition of $I$ in Eq.\ref{eq:I}, we have
\[
	X_{a}=\sigma^j\dfrac{\partial}{\partial y^j}+\phi^q_p \xi^p\dfrac{\partial}{\partial \eta^q}.
\]
Without loss of generality, assume $\mathcal{E}_{k_0}$ is a trivial bundle over $\mathcal{A}$. It can be checked that $N^\bot\mathcal{A}_{\sigma,\phi}$ has a global frame $\{\mu^j=dy^j-\dfrac{\partial \sigma^j}{\partial x^i}dx^i, \kappa^q=d\eta^q-\phi^q_p d\xi^p\}$. For an $l$-multi vector field $Z\in\mathfrak{X}(\mathcal{E})$, denote $Z(\mu^J,\kappa^K)|_{\mathcal{E}_{k_0}}$ by $Z^{JK}$ in which $J,K$ are multi-indices satisfying $|J|+|K|=l$. Notice that the action of $Z|_{\mathcal{E}_{k_0}}$ on $N^\bot\mathcal{A}_{\sigma,\phi}$ is completely determined by the collection $\{Z^{JK}\,|\,|J|+|K|=l\}$. If $l=0$, the multi-vector field $Z$ degenerates to a function. In this case, the action $Z|_{\mathcal{E}_{k_0}}(N^\bot\mathcal{A}_{\sigma,\phi})$ degenerates to $Z|_{\mathcal{E}_{k_0}}$. Technically, we have $Z|_{\mathcal{E}_{k_0}}(N^\bot\mathcal{A}_{\sigma,\phi})$ is equal to $I\circ P(Z)$, the multi-vector field determined by the multi-section $P(Z)\in \mathfrak{a}$ by the following lemma.
\begin{lemma}
\label{lem:maclaurin_expansion}
The image of the multi-vector field $Z$ under the map
\begin{eqnarray*}
	H: \mathfrak{X}(\mathcal{E}) &\to& \mathfrak{X}(\mathcal{E})\\
	Z & \mapsto &\sum_{k=0}^\infty \dfrac{1}{k!} (I\circ P) [\cdots [Z,\overbrace{-X_a]_\text{SN},\cdots, -X_a]_\text{SN}}^{k}
\end{eqnarray*}
is a formal $l$-multi-vector field tangent to the fibers of $p^!NS\oplus p^!\mathcal{B}$, in which the coefficient of the term $\dfrac{\partial}{\partial y^J}\wedge \dfrac{\partial}{\partial \eta^K}$ is the Maclaurin series of $Z^{JK}$ with respect to the linear coordinates of the fibers of $p^!NS$ and the fibers of $p^!\mathcal{B}$.
\end{lemma}
\begin{proof}
The theorem can be proved by induction on the multiplicity of $Z$. Since $H$ is linear in $Z$, we may assume $Z$ has only one summand. 
\par 1) If $Z$ is of multiplicity $0$, i.e. $Z$ is a smooth function over $\mathcal{E}$. Noticing that $P(\gamma^{r_k}_k)=0$ and $[\gamma^{r_k}_k,X_a]_\text{SN}=0$, we may assume $Z=Z(x,y,\xi,\eta)$ does not depend on $\gamma^{r_k}_k$. By straightforward calculation, we get
\[
	(I\circ P)[\cdots [Z,\overbrace{-X_a]_\text{SN},\cdots, -X_a]_\text{SN}}^{k}=\sum_{|J|+|K|=k} {k\choose |J|}\dfrac{\partial^k Z}{\partial y^J\partial \eta^K}(x,0,\xi,0)\sigma^J \phi^K,
\]
in which we adopt the abbreviation $\displaystyle \phi^q=\sum_{p}\phi^q_p\xi^p$. This indicates $H(Z)$ is the Maclaurin series of $Z|_{\mathcal{E}_{k_0}}=Z(x^i,\sigma^j(x),\xi^p,\phi^q)$ with respect to the coordinates $y^j$ and $\eta^q$, which verifies the theorem.
\par 2) If $Z$ is of multiplicity $1$, we only verify the cases $Z$ is a constant vector field in this step and prove the general cases in 3). \ding{192} If $Z=\dfrac{\partial}{\partial \gamma^{r_k}_k}$, we have $H(Z)=0$, which is consistent with $Z|_{\mathcal{E}_{k_0}}=0$. \ding{193} If $Z=\dfrac{\partial}{\partial x^i}$, the only non-vanishing term in $H(Z)$ is $(I\circ P)[Z,-X_a]_\text{SN}=-\dfrac{\partial\sigma^j}{\partial x^i}$. The lemma holds in this case if noticing that $Z|_{\mathcal{E}_{k_0}}(\mu^j)=-\dfrac{\partial \sigma^j}{\partial x^i}$ and $Z(\kappa^q)=0$. \ding{194} If $Z=\dfrac{\partial}{\partial y^j}$, we have $[\cdots [Z,\overbrace{-X_a]_\text{SN},\cdots, -X_a]_\text{SN}}^{k}=0$ when $k\ge 1$, and the only non-vanishing term is $I\circ P(Z)=Z$. This in turn matches $Z(\mu^{j'})=\delta^{jj'}$ and $Z(\kappa^q)=0$ where $\delta^{jj'}$ is the Kronecker symbol. \ding{195} The cases $Z=\dfrac{\partial}{\partial \xi^p}$ is similar to $Z=\dfrac{\partial}{\partial x^i}$ and the case $Z=\dfrac{\partial}{\partial \eta^q}$ is similar to $Z=\dfrac{\partial}{\partial y^j}$.
\par 3) Assume the lemma is true when the multiplicity of $Z$ is no greater than $N$. Now consider the case the multiplicity is $N+1$. Clearly, $Z$ can be written as a product $Z_1\wedge Z_2$ with $|Z_1|=N$, $|Z_2|=1$ and $Z_2$ being a constant vector field. By 2), we have $[\cdots [Z_2,\overbrace{-X_a]_\text{SN},\cdots, -X_a]_\text{SN}}^{k}=0$ when $k\ge 1$. As a result,  
\begin{eqnarray*}
	[\cdots [Z_1Z_2,\overbrace{-X_a]_\text{SN},\cdots, -X_a]_\text{SN}}^{k} &=& [\cdots [Z_1,\overbrace{-X_a]_\text{SN},\cdots, -X_a]_\text{SN}}^{k}\wedge Z_2 \\
    && +{k\choose 1} [\cdots [Z_1,\overbrace{-X_a]_\text{SN},\cdots, -X_a]_\text{SN}}^{k-1}\wedge [Z_2,-X_a]_\text{SN}.
\end{eqnarray*}
This gives rise to $H(Z_1Z_2)=H(Z_1)H(Z_2)$ and verifies the lemma in this case. By the principle of induction, the lemma is proved.
\end{proof}

\par The Maclaurin series in Lemma \ref{lem:maclaurin_expansion} may not be convergent in general. To ensure convergence, it is natural to require that $Z$ is analytic in $y^j$ and $\eta^q$. The following definition does not depend on local coordinates.

\begin{definition}
A function $F\in\mathcal{C}^\infty(\mathcal{E}_{k_0})$ is a \emph{polynomial} along $p^!NS$ if $G|_{p^!NS}$ is a polynomial in the linear coordinates of the fibers of $p^!NS$ with coefficients in $\mathcal{C}^\infty(\mathcal{A})$. A function $G$ is called \emph{analytic} along $p^!NS$ if $G$ has Maclaurin series with respect to the linear coordinates of the fibers of $p^!NS$ and the Maclaurin series is convergent on all fibers. An $l$-multi-vector field $Z$ on $\mathcal{E}$ is \emph{analytic} along $p^!NS$ if for any functions $F_1,\cdots, F_l\in\mathcal{C}^\infty(\mathcal{E}_{k_0})$ that are polynomial along $p^!NS$, one has $Z|_{\mathcal{E}_{k_0}}(H_1,\cdots,H_l)$ is an analytic function along $p^!NS$. A multi-vector field is analytic along $p^!NS$ if its component of multiplicity $l$ is analytic along $p^!NS$ for all $l$.
\end{definition}
\par\noindent We are ready to state the following result.
\begin{theorem}
\label{thm:coiso_deform}
Suppose $Q$ is analytic along $p^!NS$. Given $(\sigma,\phi)\mathfrak{a}'_0$, the submanifold $\mathcal{A}_{\sigma,\phi}$ is of coisotropic type if and only if $-(\sigma,\phi)$ is a Maurer-Cartan element of the $L_\infty$-algebra $\mathfrak{a}_Q^P$ in Theorem \ref{thm:l_infty_of_coiso}. 
\end{theorem}
\begin{proof}
By the analyticity of $Q$, the Maurer-Cartan equation of $-(\sigma,\phi)$ in $\mathfrak{a}^P_Q$ is convergent. Since $Q$ is polynomial with respect to the coordinates $\eta^q$, the analyticity condition of $Q$ along $p^!\mathcal{B}$ is automatically satisfied. By Lemma \ref{lem:maclaurin_expansion}, $-(\sigma,\phi)$ is a Maurer-Cartan element is equivalent to $Q|_{\mathcal{E}_{k_0}}(N^\bot\mathcal{A}_{\sigma,\phi})=0$, i.e. $\mathcal{A}_{\sigma,\phi}$ is a submanifold of coisotropic type.  
\end{proof}
\begin{remark}
If $S=M$, the normal bundle $NS$ is trivial. The analyticity condition on $Q$ is trivially satisfied in this situation.
\end{remark}
\par Generally, $\mathfrak{a}_0$ contains elements not in $\mathfrak{a}'_0$. Although a deformation must be a Maurer-Cartan element, conversely, a Maurer-Cartan element of $\mathfrak{a}_Q^P$ may not be a deformation of $\mathcal{A}$. However, this happens only when $k_0=-1$, since in this case $\dfrac{\partial}{\partial \eta^q}$ has total degree $0$ and can be multiplied to any homogeneous multi-vector field without changing the total degree. Readers shall see that the situation $k_0=-1$ dominates all our applications. Alternatively, we first find a special case in which the 1-1 correspondence can be reachieved.
\begin{corollary}
\label{cor:coiso_deform}
If $Q$ is analytic along $p^!NS$ and $Q$ is of multiplicity $1$, the $L_\infty$-algebra ${\mathfrak{a}'}_Q^{P'}$ in Theorem \ref{thm:l_infty_of_coiso2} controls the deformations of $\mathcal{A}$ so that there is 1-1 correspondence between the deformations of $\mathcal{A}$ and the Maurer-Cartan elements of ${\mathfrak{a}'}_Q^{P'}$.
\end{corollary}

\par\noindent Secondly, one many enforce the 1-1 correspondence between deformations and Maurer-Cartan elements by shrinking the space $\mathfrak{a}_0$. Let $\bar{\mathfrak{a}}$ be the graded vector space defined by
\[
	\bar{\mathfrak{a}}_0=\mathfrak{a}'_0,\quad \text{ and }\quad \bar{\mathfrak{a}}_k=\mathfrak{a}_k\text{ when }k\ne 0.
\]
Since the structure maps $\{m_k\}$ of $\mathfrak{a}$ are of degree $1$, no image under these structure maps lies in $\mathfrak{a}_0$. This means $\{m_k\}$ are closed on $\bar{\mathfrak{a}}$, and the subspace $\bar{\mathfrak{a}}$ is an $L_\infty$-subalgebra of $\mathfrak{a}$, denoted by $\bar{\mathfrak{a}}_Q^P$.
\begin{corollary}
\label{cor:coiso_deform_perfect}
If $Q$ is analytic along $p^!NS$, the $L_\infty$-algebra $\bar{\mathfrak{a}}_Q^P$ defined above controls the deformations of $\mathcal{A}$ as a submanifold of coisotropic type, i.e. there is a 1-1 correspondence between the deformations of $\mathcal{A}$ and the Maurer-Cartan elements of $\bar{\mathfrak{a}}_Q^P$.
\end{corollary}

\par Our results apply to the following situations immediately.
\begin{example}
Let $S$ be a coisotropic submanifold of a Poisson manifold $(M,\pi)$. By Theorem \ref{thm:l_infty_of_coiso}, $S$ is associated with a flat $L_\infty$-structure on $\Gamma(\wedge^\cdot NS)[1]$. By Theorem \ref{thm:coiso_deform}, if $\pi$ is analytic along $NS$, the $L_\infty$-algebra $(\Gamma(\wedge^\cdot NS)[1])_\pi^P$ controls the deformations of $S$. Thus, we recover the construction of an $L_\infty$-algebra associated with the coisotropic submanifold $S$ in \cite{cf} and the result on deformations of $S$ in \cite{sz}.
\par\noindent In a general blended $Q$-manifold $M$ of degree $0$, one has with $Q=\pi+f$. Similarly, the deformations of a submanifold $S$ of coisotropic type is controlled by the blended $L_\infty$-algebra $(\Gamma(\wedge^\cdot NS)[1])_{\pi+f}^P$ upon assuming $\pi$ is analytic along $NS$.
\end{example}

\begin{example}
Let $E\to S$ be a Lie subalgebroid of $(A,[\cdot,\cdot],\rho)$, then $E[1]$ is submanifold of $(A[1],d_A)$ of coisotropic type. Let $F$ be a complement of $E$ in $A|_S$. Since $d_A$ is of multiplicity $1$, by Theorem \ref{thm:l_infty_of_coiso2}, $E$ is associated with a flat $L_\infty$-algebra $\mathfrak{a}'=\Gamma(S(E[1]^*)\otimes (NS\oplus F[1]))$. By Corollary \ref{cor:coiso_deform}, if $d_A$ is analytic along $p^!NS$, the $L_\infty$-algebra ${\mathfrak{a}'}_{d_A}^{P'}$ controls the deformations of $E$. Here $p^!NS$ is the pull-back of $NS$ along $E\to S$. This result is first found by the author in \cite{j}, and we re-obtain it from a new viewpoint. We also point out that when $S=M$, the analyticity condition is not necessary. The space $\mathfrak{a}'$ degenerates to $S(E[1]^*)\otimes F[1]$, and only the first three higher derived brackets of $d_A$ are non-vanishing.
\end{example}

\subsection{Simultaneous Deformations}
\par The $L_\infty$-algebras $\mathfrak{a}_Q^P$ and ${\mathfrak{a}'}_Q^{P'}$ are constructed from the V-algebras $(\mathfrak{X}(\mathcal{E}),[\cdot,\cdot]_\text{SN},\mathfrak{a},P)$ and $(\mathfrak{X}^1(\mathcal{E}),[\cdot,\cdot]_\text{SN},\mathfrak{a},P)$ ($Q$ is of multiplicity $1$ in the second case), respectively. We can apply Theorem \ref{thm:cs} if $Q$ is positive (i.e. the degree $-1$ component of $Q$ vanishes).
\begin{theorem}
\label{thm:simul_deform}
If $Q$ is of degree $1$ and analytic along $p^!NS$, there is a flat $L_\infty$-structure on $\mathfrak{X}(\mathcal{E})[1]\oplus \mathfrak{a}$ whose structure maps are given by Eq.\ref{eq:comb_l_inf1} -- \ref{eq:comb_l_inf5}. The $L_\infty$-algebra $(\mathfrak{X}(\mathcal{E})[1]\oplus \mathfrak{a})_Q^P$ controls the simultaneous deformations of $Q$ and $\mathcal{A}$ in the sense that given $\tilde{Q}\in \mathfrak{X}(\mathcal{E})$ and $(\sigma,\phi)\in\mathfrak{a}'_0$ satisfying $\tilde{Q}$ is analytic along $p^!NS$ and $\|\tilde{Q}\|=1$, then $(\tilde{Q},-(\sigma,\phi))$ is a Maurer-Cartan element of $(\mathfrak{X}(\mathcal{E})[1]\oplus \mathfrak{a})_Q^P$ if and only if
\[
\begin{cases}
Q+\tilde{Q} \text{ is a homological vector field of } \mathcal{E}, \text{ and}\\
\mathcal{A}_{\sigma,\phi} \text{ is a submanifold of coisotropic type w.r.t. } Q+\tilde{Q}.
\end{cases}
\]
\end{theorem}
\par\noindent Parallel to Corollary \ref{cor:coiso_deform}, an $L_\infty$-subalgebra may be constructed.
\begin{corollary}
\label{cor:simul_deform}
Suppose $Q$ satisfies the hypothesis of Theorem \ref{thm:simul_deform}, and furthermore $Q$ has multiplicity $1$. The $L_\infty$-algebra structure on $\mathfrak{X}(\mathcal{E})[1]\oplus \mathfrak{a}$ can be restricted to $\mathfrak{X}^1(\mathcal{E})[1]\oplus \mathfrak{a}'$ to get an $L_\infty$-subalgebra that controls the simultaneous deformations of $Q$ and $\mathcal{A}$.
\end{corollary}

\par The result on simultaneous deformations of a Poisson structure and a coisotropic submanifold is first presented in \cite{fz}, while that of a Lie algebroid and its Lie subalgebroid is obtained in \cite{j}. By our viewpoint, both results can be recovered easily.
\begin{example}
Given a coisotropic submanifold $S$ of a Poisson manifold $(M,\pi)$, by Theorem \ref{thm:simul_deform}, there is a flat $L_\infty$-structure on $\mathfrak{X}(M)[1]\oplus \Gamma(\wedge^\cdot NS)$. If $\pi$ is analytic along $NS$, then this $L_\infty$-algebra controls the simultaneous deformations of the Poisson bi-vector field $\pi$ and the coisotropic submanifold $S$. That is, given $(\tilde{\pi},\sigma)\in \mathfrak{X}^2(M)\oplus \Gamma(NS)$, assuming $\tilde{\pi}$ is analytic along $NS$, the pair $(\tilde{\pi},-\sigma)$ is a Maurer-Cartan element of $(\mathfrak{X}(M)[1]\oplus \Gamma(\wedge^\cdot NS))_{\pi}^P$ if and only if
\[
\begin{cases}
Q+\tilde{Q} \text{ is a Poisson bi-vector field on } M, \text{ and}\\
\text{graph}(\sigma) \text{ is a coisotropic submanifold of } M \text{ w.r.t. the Poisson structure } \pi+\tilde{\pi}.
\end{cases}
\]
\end{example}

\begin{example}
Let $E\to S$ be a Lie subalgebroid of $(A,[\cdot,\cdot],\rho)$, and $F$ a complement of $E$ in $A|_S$ as before. If $d_A$ is analytic along $p^!NS$, by Corollary \ref{cor:simul_deform}, there is a flat $L_\infty$-structure on $\mathfrak{X}^1(A[1])[1]\oplus\mathfrak{a}'$ that controls the simultaneous deformations of the Lie algebroid structure $([\cdot,\cdot],\rho)$ and the Lie subalgebroid $E$. To be specific, given $(\tilde{d}_A,(\sigma,\phi))\in (\mathfrak{X}^1(A[1])[1]\oplus\mathfrak{a}')_0$, as long as $\tilde{d}_A$ is analytic along $p^!NS$, the pair $(\tilde{d}_A,-(\sigma,\phi))$ is a Maurer-Cartan element of $\mathfrak{X}^1{A[1]}\oplus\mathfrak{a}'$ if and only if 
\[
\begin{cases}
d_A+\tilde{d}_A \text{ defines a Lie algebroid structure on } A, \text{ and}\\
\text{graph}(\sigma,\phi) \text{ is a Lie subalgebroid of } A \text{ w.r.t. the Lie algebroid structure defied by }d_A+\tilde{d}_A.
\end{cases}
\]
As before, the analyticity condition on $d_A$ can be omitted if $S=M$.
\end{example}

\section{Application: Deformations of Courant Algebroid and Dirac Structure}
\label{sec:app}
\par By D. Roytenberg, a Courant algebroid is a symplectic $Q$-manifold of degree $2$. Translating to our language, a Courant algebroid is a blended $Q$-manifold of degree $2$. This allows us to apply the results in Section \ref{sec:deform} to attack the deformation problems in Courant algebroid. Courant algebroid was first defined in \cite{lxw}. Thanks to \cite{u}, the axioms in the definition got simplified. Nowadays, a Courant algebroid is defined as follows.

\begin{definition}
A \emph{Courant algebroid} is a vector bundle $E\to M$ together with a non-degenerate $\mathcal{C}^\infty(M)$-bilinear form $\langle\cdot,\cdot\rangle$ on $\Gamma(E)$ (called a pseudo-inner product), a $\mathbb{R}$-bilinear operation on $\Gamma(E)$ (called the Dorfman bracket), and a bundle map $\rho: E\to TM$ (call the anchor), satisfying
\par 1) $e_1\circ(e_2\circ e_3)=(e_1\circ e_2)\circ e_3+e_2\circ(e_1\circ e_3)$,
\par 2) $\rho(e)\langle e_1,e_2\rangle =\langle e\circ e_1, e_2\rangle+\langle e_1,e\circ e_2\rangle$, and
\par 3) $e_1\circ e_2+e_2\circ e_1=D\langle e_1,e_2\rangle$,
\par\noindent for any $e,e_1,e_2,e_3\in\Gamma(A)$. Here $D:\mathcal{C}^\infty(M)\to \Gamma(E)$ is defined by $D=\rho^*d$, i.e. $\langle Df, e\rangle =\rho(e)f$ for any $f\in\mathcal{C}^\infty(M)$.
\end{definition}
\par\noindent The first two axioms in the definition requires $e$ being a derivation with respect to both $\circ$ and $\langle\cdot,\cdot\rangle$. In general, the Dorfman bracket $\circ$ is not anti-symmetric and the last axiom measure how far for $\circ$ being anti-symmetric. Instead the Dorfman bracket, one may adopt the bracket
\[
	[e_1, e_2]=\dfrac{1}{2}(e_1\circ e_2-e_2\circ e_1)
\]
in the definition. The discussions on the axioms required for this bracket and the equivalence of the two definitions can be found in \cite{r1}. Furthermore, the nondegeneracy of $\langle\cdot,\cdot\rangle$ forces the rank of $E$ to be an even number.

\begin{remark} 
From the definition, the following properties can be derived:
\par 1) $\rho(e_1\circ e_2)=[\rho(e_1),\rho(e_2)]_\text{SN}$,
\par 2) $\rho(e_1\circ(fe_2))=f[\rho(e_1),\rho(e_2)]_\text{SN}+(\rho(e_1)f)e_2$, 
\par 3) $\langle Df,Dg \rangle=0$ ($\rho\circ D=0$),
\par 4) $e\circ (\rho^*f)=\rho^*\mathcal{L}_{\rho(e)}(f)$,
\par 5) $(\rho^*f)\circ e=-\rho^*(\iota_{\rho(e)} df)$.
\par\noindent Here, $\mathcal{L}_u$ and $\iota_u$ are the Lie derivative and contraction of a vector field $u\in\mathfrak{X}(M)$, respectively. In the last two identities, $\rho^*f$ is identified with a section of $E$ by the isomorphism $E\to E^*$ defined by $e\to \langle e,\cdot\rangle$. 
\end{remark}

\par Next, let us briefly describe the 1-1 correspondence between Courant algebroids and symplectic $Q$-manifolds of degree $2$ found in \cite{r2}. Given a symplectic $Q$-manifolds $(\mathcal{E},\Omega)$ with $\Omega$ the symplectic structure satisfying $|\Omega|=2$, by the non-degeneracy of $\Omega$, the degree of the graded manifold $\mathcal{E}$ can not exceed $2$. Consequently, $\mathcal{E}$ is concentrated in degree $-1$ and $-2$ only. Moreover, the Poisson structure $\pi$ inverse to $\Omega$ is non-degenerate and of degree $|\pi|=-2$. Let $\mathcal{A}^k$ be the collection of homogeneous functions on $\mathcal{E}$ of degree $k$. Particularly, we have $\mathcal{A}^0=\mathcal{C}^\infty(M)$, and $\mathcal{A}^1=\Gamma(\mathcal{E}_{-1}^*)$. Reader will find that the Courant algebroid structure is on $\mathcal{E}^*_{-1}$. For simplicity, let us denote it by $E$.

\par 1)  Since $\pi(\mathcal{A}^1,\mathcal{A}^1)\subset \mathcal{A}^0$, $\pi(\mathcal{A}^1,\mathcal{A}^0)=0$ and $\pi(\mathcal{A}^0,\mathcal{A}^0)=0$, the Poisson bi-vector field $\pi$ induces a pseudo-inner product $\langle\cdot,\cdot\rangle=\pi(\cdot,\cdot)$ on $E=\mathcal{E}^*_{1}$.

\par 2) Since $\mathcal{A}^2$ is a locally free sheaf of $\mathcal{A}^0$-modules, there exists a vector bundle $V\to M$ so that $\mathcal{A}^2=\Gamma(V)$. The relations $\pi(\mathcal{A}^2,\mathcal{A}^2)\subset \mathcal{A}^2$ and $\pi(\mathcal{A}^2,\mathcal{A}^0)\subset \mathcal{A}^0$ together with the Leibniz property defines Lie algebroid structure $([\cdot,\cdot],a)$ on $V$. Furthermore, the relation $\pi(\mathcal{A}^2,\mathcal{A}^1)\subset \mathcal{A}^1$ together with the non-degeneracy of $\pi$ implies $V$ is the Gauge Lie algebroid of $(\mathcal{E}_{-1}^*,\langle\cdot,\cdot\rangle)$, i.e. sections of $V$ act on $\mathcal{A}^1$ as covariant differential operators that preserves $\langle\cdot,\cdot\rangle$. To see this, first, the anchor $a$ is surjective by the nondegeneracy; second, $\mathcal{A}^1\mathcal{A}^1=\Gamma(\wedge^2 \mathcal{E}_{-1}^*)\subset \mathcal{A}^2$ is contained in the kernel of $a$ since it acts on $\mathcal{A}^1$ trivially; third, by the Leibniz rule and the Jacobi identity, the action of $\mathcal{A}^2$ on $\mathcal{A}^1$ is a Lie algebroid action that preserves $\langle\cdot,\cdot\rangle$, so there is a morphism from $V$ to the gauge Lie algebroid of $\mathcal{E}_{-1}^*$ which is an isomorphism. Therefore, there is a short exact sequence
\[
	0\to \wedge^2\mathcal{E}_{-1}^*\to V \xrightarrow{a} TM\to 0,
\]
and one has $\mathcal{E}_{-2}=T[2]M$ (the fibers of $TM$ has degree $-2$). Together with $\mathcal{E}_{-1}=E^*[1]$ and the fact that the Poisson structure $\pi$ can be recovered from $\langle\cdot,\cdot\rangle$, the $Q$-manifold $(\mathcal{E},\Omega)$ is completely determined by $(E,\langle\cdot,\cdot\rangle)$. 

\par 3) Any Poisson vector field of degree $|\cdot|$ greater that $-2$ must be Hamiltonian. Since a homological vector field on $(\mathcal{E},\pi)$ is required to preserve $\pi$, it must be of the form $X_\theta=\iota_\Theta(\pi)$ with $\Theta\in\mathcal{C}^\infty(\mathcal{E})$ a cubic function. The cubic function $\Theta$ induces the maps $\circ: \Gamma(E)\times\Gamma(E)\to \Gamma(E)$ and $\rho:\Gamma(E)\to TM$ as derived brackets
\begin{eqnarray*}
	e_1\circ e_2 &=& \{\{e_1,\Theta\},e_2\},\qquad \forall e_1,e_2\in\Gamma(E)\\
    \rho(e) f &=& \{\{e,\theta\},f\},\qquad \forall e\in\Gamma(E) \text{ and } f\in\mathcal{C}^\infty(M). 
\end{eqnarray*}
It is proved that $\{\Theta,\Theta\}=0$ is equivalent to $(E,\langle\cdot,\cdot\rangle,\circ,\rho)$ is a Courant algebroid. Here $\{\cdot,\cdot\}$ is the Poisson bracket of $\pi$. 
\par As a result, a Courant algebroid structure on $E$ is determined by the pair $(\pi,X_\Theta)$. The identity $\{\Theta,\Theta\}=0$ is equivalent to
\[
	[X_\Theta,X_\Theta]_\text{SN}=0.
\]
Together with the relations
\[
	[\pi,\pi]_\text{SN}=0,\qquad [\pi,X_\Theta]_\text{SN}=0,
\]
and the fact that $\|\pi\|=-1$ and $\|X_\Theta\|=1$, we conclude the following lemma.
\begin{lemma}
\label{lem:courant_q}
A Courant algebroid structure $(\langle\cdot,\cdot\rangle,\circ,\rho)$ on $E$ is equivalent to the blended $Q$-vector field $\pi+X_\Theta$ on $\mathcal{E}$.
\end{lemma}
\par\noindent Then we can apply Theorem \ref{thm:q_deform}.
\begin{theorem}
The deformations of a Courant algebroid $(E,\langle\cdot,\cdot\rangle,\circ,\rho)$ is controlled by the blended Lie algebra $(\mathfrak{X}(\mathcal{E}),[\cdot,\cdot]_\text{SN},[\pi+X_\Theta,\cdot]_\text{SN})$, i.e. given $\tilde{\pi}\in\mathfrak{X}^2(\mathcal{E})$ and $\tilde{X}\in\mathfrak{X}^1(\mathcal{E})$ so that $|\tilde{\pi}|=-2$, $|\tilde{X}|=1$ and $\pi+\tilde{\pi}$ is non-degenerate, the pair $(\pi+\tilde{\pi},X_\Theta+\tilde{X})$ defines a Courant algebroid structure on $E$ if and only if $(\tilde{\pi},\tilde{X})$ satisfies the Maurer-Cartan equation
\[
	[\pi+X_\Theta,\tilde{\pi}+\tilde{X}]_\text{SN}+\frac{1}{2}[\tilde{\pi}+\tilde{X},\tilde{\pi}+\tilde{X}]_\text{SN}=0.
\]
\end{theorem}
\par Let $(q^i)$ be a local coordinate system of $M$ and $\{\xi^a\}$ a local frame of $E$ so that $g^{ab}=\langle \xi^a,\xi^b\rangle$ are constant. Then $(q^i,\xi,p_i)$ forms a Darboux chart of $\mathcal{E}$ in which $(q^i,p^i)$ is the standard Darboux chart of $\mathcal{E}_{-2}=T[2]M$ and $(\xi^a)$ are viewed as linear coordinates of the fibers of $\mathcal{E}_{-1}=E^*[1]$. Locally, $\pi$ is given by
\[
	\pi = \dfrac{\partial}{\partial q^i}\dfrac{\partial}{\partial p_i}+\dfrac{1}{2}\dfrac{\partial}{\partial \xi^a} g^{ab}\dfrac{\partial}{\partial \xi^b}.
\]
Only the second summand of $\pi$ affects the pseudo-inner product, which is the image of $\pi$ under the natural projection $\mathfrak{X}(\mathcal{E})\to \mathfrak{X}(\mathcal{E}_{-1})$. We remark that to deform the pseudo-inner product $\langle\cdot,\cdot\rangle$, one may restrict the change to being in the second summand only.
\par\noindent Besides, a cubic $\Theta$ can be written as 
\[
	\Theta=f_a^i \xi^a p_i-\dfrac{1}{6}h_{abc}\xi^a\xi^b\xi^c,
\]
with $f_a^i,h_{abc}\in\mathcal{C}^\infty(M)$ determined by
\begin{eqnarray*}
	g^{ab}f_b^i &=& \rho(\xi^a)x^i,\\
    g^{aa'}g^{bb'}g^{cc'}h_{a'b'c'} &=& \langle \xi^a\circ \xi^b, \xi^c\rangle.
\end{eqnarray*}
\par Next, let us consider the deformations of a Dirac structure.
\begin{definition}
Given a Courant algebroid $E$, a subbundle $A$ over $M$ is called a Dirac structure or Dirac subbundle if it is maximally isotropic under the pseudo-inner product $\langle\cdot,\cdot\rangle$ and its sections are closed under the Dorfman bracket $\circ$.   
\end{definition}
\par\noindent  Here, $A$ is isotropic requires $\langle \xi^a,\xi^b\rangle=0$ for any $\xi^a,\xi^b\in \Gamma(A)$. The maximal isotropy of $A$ indicates the rank of $A$ is half of that of $E$. Furthermore, $A$ is a Dirac structure implies $(A,\circ,\rho)$ is a Lie algebroid.
\par Given a subbundle $A\to M$ of $E$, via the pseudo-inner product, $A$ can be identified with a subbundle of $E^*$, denoted by $\mathcal{B}$.
\begin{lemma}
\label{lem:dirac_q}
A subbundle $A$ of rank $\frac{1}{2}$rk$(E)$ is a Dirac structure if and only if $\mathcal{B}[1]\subset\mathcal{E}$ is of coisotropic type with respect to $\pi+X_\Theta$.
\end{lemma}
\begin{proof}
Let $B$ be a complement of $A$, then one may identify $\mathcal{B}$ with $B^*$. Suppose $E$ has rank $2n$, and let $\{\xi^1,\cdots, \xi^n\}$, $\{\xi^{n+1},\cdots,\xi^{2n}\}$ be local frames of $A$ and $B$ respectively. It is easy to see $N^\bot (\mathcal{A}[1])$ has $\{d\xi^1,\cdots, d\xi^n\}$ as a local frame. It follows that $\mathcal{B}[1]$ is of coisotropic type with respect to $\pi$, i.e. $\pi|_{\mathcal{E}_{-1}}(N^\bot \mathcal{B}[1])=0$, if and only if $g^{ab}=0$ when $1\le a,b\le n$, but this is equivalent to that the subbundle $A$ is isotropic under $\langle\cdot,\cdot\rangle$.

\par Locally, $X_\Theta|_{\mathcal{E}_{-1}}$ has the form
\begin{eqnarray*}
	\dfrac{1}{2}h_{abc}g^{cd}\xi^a\xi^b\dfrac{\partial}{\partial \xi^d}-\dfrac{1}{6}\dfrac{\partial h_{abc}}{\partial q^i}\xi^a\xi^b\xi^c\dfrac{\partial}{\partial p_i}+\xi^af_a^i\dfrac{\partial}{\partial q^i}.
\end{eqnarray*}
Then $X_\Theta|_{\mathcal{E}_{-1}}(N^\bot(\mathcal{A}[1]))=0$ is equivalent to 
\begin{eqnarray}
	h_{abc}g^{cc'}=0\qquad\text{ when }a,b>n\text{ and }c'\le n.\label{eq:first}
\end{eqnarray}
Meanwhile, $\Gamma(A)$ is closed under $\circ$ if $\xi^{a'}\circ \xi^{b'}\in \Gamma(A)$ whenever $a',b'\le n$, i.e. 
\begin{eqnarray}
	h_{abc}g^{aa'}g^{bb'}=0\qquad\text{ when }a',b'\le n\text{ and }c>n.\label{eq:second}
\end{eqnarray}
It turns out the two conditions are equivalent given that $A$ is isotropic. The reason is as follows: given smooth functions $k_1,\cdots, k_{2n}\in\mathcal{C}^\infty(M)$ so that $\sum_{a=1}^{2n}k_ag^{aa'}=0$ when $a'\le n$, it is equivalent to $\sum_{a=n+1}^{2n}k_ag^{aa'}=0$ since $g^{aa'}=0$ if $a,a'\le n$; then by the non-degeneracy of $\langle\cdot,\cdot\rangle$, one has the matrix $(g^{aa'})_{\substack{n<a<2n\\1\le a'\le n}}$ is invertible; as a result, $\sum k_ag^{aa'}=0$ is equivalent to $k_a=0$ for $a>n$. Applying this property to the $c$-component in $h_{abc}g^{cc'}$ and the $a,b$-components in $h_{abc}g^{aa'}g^{bb'}$, we get Eq.\ref{eq:first} and Eq.\ref{eq:second} are equivalent.
\par Combining the two aspects together, $A$ is a Dirac structure if and only if $\mathcal{A}[1]\subset\mathcal{E}$ is of coisotropic type with respect to $\pi+X_\Theta$.
\end{proof}
\par\noindent This allows us to apply Theorem \ref{thm:l_infty_of_coiso} and Corollary \ref{cor:coiso_deform_perfect}.
\begin{theorem}
With a selection of the complement bundle $B$ of the Dirac structure $A$ in a Courant algebroid $E$, there is a flat $L_\infty$-structure on $\mathfrak{a}=\Gamma(S(A^*)\otimes S(B))$, so that the deformations of $A$ as a Dirac structure is controlled by its $L_\infty$-subalgebra $\bar{\mathfrak{a}}$ in Corollary \ref{cor:coiso_deform_perfect}, i.e. the graph of $\phi:A\to B$ is a Dirac structure if and only if $\phi\in\bar{\mathfrak{a}}_0$ is a Maurer-Cartan element of $\bar{\mathfrak{a}}$. Here $\bar{\mathfrak{a}}$ is the graded vector space defined by $\bar{\mathfrak{a}}_0=\Gamma(A^*\otimes B)$ and $\bar{\mathfrak{a}}_k=\Gamma(S(A^*)\otimes S(B))_k$ when $k\ne 0$.
\end{theorem}
\par\noindent Since $\pi+X_\Theta$ is a blended homological vector field, Theorem \ref{thm:cs} is not applicable to control the simultaneous deformations of a Courant algebroid $E$ and its Dirac structure $A$. Besides, a Dirac structure $A$ is over the whole base manifold, the Maurer-Cartan equation of $\bar{\mathfrak{a}}$ is convergent without requiring any analyticity condition.


\begin{thebibliography}{99}
\addcontentsline{toc}{section}{Bibliography}
\frenchspacing
\bibitem{b}
M. Batchelor, \emph{The structure of supermanifolds}, Transactions of the American Mathematical Society, 253 (1979):329-338. 

\bibitem{c}
T. Courant, \emph{Dirac Manifolds}, Transactions of the American Mathematical Society, 319(2) (1990):631-661.
\bibitem{cf}
A. S. Cattaneo and G. Felder,
\emph{Relative formality theorem and quantization of coisotropic submanifolds}, Advances in Mathematics, 208 (2007):521-548.

\bibitem{cs}
A. S. Cattaneo and F\mbox{}. Sch\"{a}tz, \emph{Equivalences of higher derived brackets}, Journal of Pure and Applied Algebra, 212(11) (2008):2450-2460.

\bibitem{fz}
Y. Fr\'egier and M. Zambon, \emph{Simultaneous deformations and Poisson geometry}, Compositio Math. 151, (2015):1763–1790.

\bibitem{gms}
M. Gualtieri, M. Matviichuk and G. Scott, \emph{Deformations of Dirac structures via $L_\infty$ algebra}, arXiv:1702.08837.

\bibitem{j}
X. Ji, \emph{Simultaneous deformations of a Lie algebroid and its Lie subalgebroid}, Journal of Geometry and Physics, 84 (2014):8-29.

\bibitem{kw1}
F. Keller and S. Waldmann, \emph{Formal deformations of Dirac structures}, Journal of Geometry and Physics, 57(3) (2007):1015–1036.

\bibitem{kw2}
F. Keller and S. Waldmann, \emph{Deformation theory of Courant algebroids via the Rothstein algebra}, Journal of Pure and Applied Algebra, 219(8) (2015):3391–3426.

\bibitem{ls}
T. Lada and J. Stasheff, \emph{Introduction to sh Lie algebras for physicists}, International Journal of Theoretical Physics, 32 (1993):1087-1104.

\bibitem{lwx}
Z. Liu, A. Weinstein and P. Xu, \emph{Manin triples for Lie bialgebroids}, Journal of Differential Geometry, 45 (1997):547-574.

\bibitem{l}
D. Li-Bland, \emph{$\mathcal{LA}$-Courant algebroids and their applications}, Ph.D. thesis, University of Toronto (2012).

\bibitem{r1}
D. Roytenburg, \emph{Courant algebroids, derived brackets and even symplectic supermanifolds}, Ph.D. thesis, UC Berkeley (1999).

\bibitem{r2}
D. Roytenburg, \emph{On the structure of graded symplectic supermanifolds and Courant algebroids}, Contemporary Mathematics, 315 (2002): 169-185.

\bibitem{sz}
F. Sch\"{a}tz and M. Zambon, \emph{Deformations of coisotropic submanifolds for fibrewise entire Poisson structures}, Letters in Mathematical Physics, 103(7) (2013):777-791.

\bibitem{u}
K. Uchino, \emph{Remarks on the Definition of a Courant Algebroid}, Letters in Mathematical Physics, 60(2) (2002):171-175.

\bibitem{v1}
Th. Voronov, \emph{Higher derived brackets and homotopy algebras}, Journal of Pure and Applied Algebra, 202(1-3) (2005):133-153.

\bibitem{v2}
Th. Voronov, \emph{Higher derived brackets for arbitrary derivations}, Travaux Math\'{e}matiques, XVI (2005):163-186.
\end{thebibliography}
\end{document}